\documentclass[article,onecolumn]{IEEEtran}

\usepackage{graphicx,colordvi,psfrag}
\usepackage{amsmath,amssymb}
\usepackage{epstopdf}
\usepackage{epsfig}
\usepackage{calc,pstricks, pgf, xcolor}
\usepackage{bbding}
\usepackage{bbm}
\usepackage{dsfont}

\newcommand{\Ind}{\mathds{1}}

\newcommand{\ZZ}{\mathbb{Z}}
\newcommand{\RR}{\mathbb{R}}

\newcommand{\Var}{\mathrm{Var}}

\newcommand{\Ber}{\mathop{\mathrm{Ber}}}

\newcommand{\Unif}{\mathrm{Uniform}}
\newcommand{\Bin}{\mathrm{Binomial}}

\newcommand\indep{\protect\mathpalette{\protect\independenT}{\perp}}
\def\independenT#1#2{\mathrel{\rlap{$#1#2$}\mkern2mu{#1#2}}}
\newcommand{\m}{\mathcal}

\newtheorem{theorem}{Theorem}

\newtheorem{proposition}{Proposition}
\newtheorem{corollary}{Corollary}

\newtheorem{definition}{Definition}
\newtheorem{lemma}{Lemma}

\newenvironment{proof}[1][Proof]{\noindent\textbf{#1.} }{\ \rule{0.5em}{0.5em}}

\begin{document}
\title{A Lower Bound on the Expected Distortion of Joint Source-Channel Coding}

\author{Yuval Kochman, Or Ordentlich and Yury Polyanskiy
	\thanks{}
	\thanks{Y. Kochman and O. Ordentlich are with the School of Computer Science and Engineering, Hebrew University of Jerusalem, Israel (emails: yuval.kochman@mail.huji.ac.il, or.ordentlich@mail.huji.ac.il).  Y. Polyanskiy is with the Massachusetts Institute of Technology, MA, USA (email: yp@mit.edu)}
	\thanks{The work of Y. Kochman was supported by the ISF under Grant 1555/18 and by the HUJI Cyber Security Research Center in conjunction with the Israel National Cyber Bureau in the Prime Minister's Office. The work of O. Ordentlich was supported by the ISF under Grant 1791/17. The work of Y. Polyanskiy was supported, in part, by the  Center for Science of Information (CSoI),
		and the NSF Science and Technology Center, under grant agreement CCF-09-39370.
		The material in this paper was presented in part at the 2018 and 2019 International Symposium on Information Theory~\cite{kop18,kop19isit}.}
	\thanks{}}

\date{}

\maketitle

\begin{abstract}
We consider the classic joint source-channel coding problem of transmitting a memoryless source over a memoryless channel. The focus of this work is on the long-standing open problem of finding the rate of convergence of the smallest attainable expected distortion to its asymptotic value, as a function of blocklength $n$. Our main result is that in general the convergence rate is not faster than $n^{-1/2}$. In particular, we show that for the problem of transmitting i.i.d uniform bits over a binary symmetric channels with Hamming distortion, the smallest attainable distortion (bit error rate) is at least $\Omega(n^{-1/2})$ above the asymptotic value, if the ``bandwidth expansion ratio'' is above $1$. 
\end{abstract}

\section{Introduction}
\label{sec: intro}

Over the last decade there has been a great progress in understanding the rate of convergence to the asymptotic fundamental limits in various communication and compression setups. Yet, there remain some setups where although the asymptotic limits are known, the rate of convergence to those is not known. Two prominent examples are to show that the
Gaussian-type $1/\sqrt{n}$ backoff, common to many settings, also arises in the  joint source-channel coding (JSCC) setup under expected
distortion, and in the multiple access channel (MAC). The main difficulty in the JSCC setup is
that due to averaging, the Gaussian variations in source/channel
quality may possibly be canceled out (as in fact happens when the `not to code' conditions of~\cite{gastpar03} are met). The
fundamental issue in the MAC setup~\cite{ap15} is that the multi-user
interference, unless randomly-coded, should not in general satisfy
the central limit theorem and hence there is no reason to believe the back-off should
be of $1/\sqrt{n}$ order. This paper makes progress on the first of these open problems.

Specifically, we consider the classical point-to-point joint source-channel coding problem, depicted in  Figure~\ref{fig:jsccgen}. In this setup, an encoder observes a sequence $S^m=(S_1,\ldots,S_m)$ of $m$ i.i.d. samples generated according to the distribution $P_S$, and would like to send this sequence through $n$ channel uses of the memoryless channel $Q_{Y|X}$. To that end, the encoder maps the source sequence $S^m$ to the channel input 
$X^n$ using an encoding function $\m{E}:\m{S}^m\to\m{X}^n$. The channel input $X^n$ is transmitted through the channel $Q^{\otimes n}_{Y^n|X^n}(y^n|x^n)=\prod_{i=1}^n Q_{Y|X}(y_i|x_i)$ and the decoder that observes the channel output $Y^n$, generates an estimate $\hat{S}^m=(\hat{S}_1,\ldots,\hat{S}_m)$ of the source sequence, using a decoding function $\m{D}:\m{Y}^n\to \hat{\m{S}}^m$. Let $d:\m{S}\times\hat{\m{S}}\to\RR$ be some distortion measure, and define
\begin{align}
d(S^m,\hat{S}^m)=\sum_{i=1}^m d(S_i,\hat{S}_i).
\end{align}

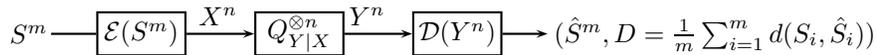
\begin{figure}[t]
	\begin{center}
		\psset{unit=0.6mm}
		\begin{pspicture}(0,50)(250,70)
		\rput(25,0){
			\rput(0,60){$S^m$}\psline(5,60)(15,60)\psframe(15,65)(35,55)\rput(25,60){$\m{E}(S^m)$}
			\psline{->}(35,60)(50,60)\rput(42,64){$X^n$}\psframe(50,65)(70,55)\rput(60,60){$Q^{\otimes n}_{Y|X}$}
			\psline{->}(70,60)(85,60)\rput(75,64){$Y^n$}\psframe(85,65)(105,55)\rput(95,60){$\m{D}(Y^n)$}\psline{->}(105,60)(115,60)\rput(152,60){$(\hat{S}^m,D=\frac{1}{m}\sum_{i=1}^md(S_i,\hat{S}_i))$}
		}
		\end{pspicture}
	\end{center}
\caption{The $(n,\rho,P_S,Q_{Y|X})$ joint source channel coding problem. It is assumed that $n=\rho m$.}
\label{fig:jsccgen}
\end{figure}

For a given source-channel pair, one is interested in the statistics of the distortion $d(S^m,\hat{S}^m)$ that may be obtained, as a function of the blocklengths $m$ and $n$. It is convenient to think of the \emph{bandwidth expansion ratio} $\rho = n/m$ as fixed (ignoring rounding effects), and then consider the performance as a function of $n$. As the full statistics of the distortion are complicated, usually one of two figures of merit is considered: the \emph{expected} distortion \[ D=\frac{1}{m}\mathbb{E} d(S^m,\hat{S}^m), \] or the excess-distortion probability, which for any threshold $D$ is given by
\[ \epsilon(D) = \Pr \left[\frac{1}{m} d(S^m,\hat{S}^m) > D \right] . \]

The focus of this work is the expected distortion, for which we define the fundamental limit for the JSCC problem by the function
\begin{align}
D^*_n = D^*(n,\rho,P_S,Q_{Y|X})\triangleq \frac{1}{m}\min_{\substack{{\m{E}:\m{S}^m\to\m{X}^n}\\{\m{D}:\m{Y}^n\to\hat{\m{S}}^m}}}\mathbb{E}d(S^m,\hat{S}^m).
\end{align}
By the separation principle \cite{coverthomas}, we have that
\begin{align}
D^*_\infty =D_{\infty}^*(\rho,P_S,Q_{Y|X})\triangleq \lim_{n\to\infty}D^*(n,\rho,P_S,Q_{Y|X})= D_{P_S}(\rho C(Q_{Y|X})),
\end{align}
where $D_{P_S}(R)$ is the distortion-rate function of a source with 
distribution $P_S$, and $C(Q_{Y|X})$ is the capacity of the channel $Q_{Y|X}$. We study the convergence of the expected distortion to its asymptotic value. To that end, we define and study the quantity
\begin{align}
\Delta^*_n=\Delta^*(n,\rho,P_S,Q_{Y|X})\triangleq D^*_n-D_\infty^*.
\end{align}

In terms of excess distortion, many facts are known about the convergence of the distortion to the infinite-blocklength limit. For any threshold $D > D^*_\infty$ we have that $\epsilon(D)$ is exponentially small, with upper and lower bounds on the exponents (which agree for low enough $D$) given in~\cite{CsiszarJointExponent}. To the contrary, for any $D <  D^*_\infty$ it holds that $1-\epsilon(D)$ is exponentially small \cite{JSCC_StrongConverse}. The dispersion, i.e., convergence of $D$ to $D^*_\infty$ for fixed excess-distortion probability, as well as finite-blocklength bounds, were derived in~\cite{JSCC_Disp_Allerton,KostinaVerduJSCC}.


In problems where
the error criterion is defined as a hard-constrained $0/1$-loss
(such as excess-distortion above), the Gaussian variations in channel
quality cannot be leveraged and the fundamental limit experiences a
$1/\sqrt{n}$ back-off from its asymptotic values. However, in 
problems with averaging, the dispersion term disappears, for example~\cite{ppv11,ppv11b,ydkp14}, because these
(mean-zero) Gaussian variations can be averaged out. If indeed the variations canceled out in a similar way for average distortion, it would suggest
that the true behavior of the fundamental limit $\Delta_n^*$ should
indeed be $o(1/\sqrt{n})$.  

As further evidence that $\Delta_n^*$ may be $o(1/\sqrt{n})$, consider the special case of the JSCC where $Q_{Y|X}$ is a clean bit-pipe of rate $R$, for which the problem reduces to lossy source coding. In this case, it is known~\cite{ZhangYangWei97,YuSpeedCovering} that for any discrete source and rate $R>0$,
\[ \Delta^*_n = \m{O}\left(\frac{\log n}{n} \right). \]

Furthermore, for some source-channel pairs the optimal asymptotic distortion is already achievable using a scalar scheme. See~\cite{gastpar03} for necessary and sufficient conditions. For example, this is the case for the problem of sending a binary symmetric source (BSS) over a binary symmetric channel (BSC) under expected Hamming distortion with $\rho=1$.
In light of this, one might hope that a low redundancy is possible in general. 

Despite all this evidence for sub-$\sqrt{n}$ convergence, this work proves that it is not the case, by showing that there exist cases where $\Delta^*_n=\Omega\left(\frac{1}{\sqrt{n}}\right)$. More concretely, we study the very same symmetric binary-Hamming problem mentioned above, but with $\rho>1$, and derive a lower bound on $D^*_n$. 


Our approach to proving this result goes through a reduction to a JSCC broadcast problem. Let $\hat{Q}_{Y|X}$ be the empirical channel realization in the point-to-point JSCC problem. The main observation in our distortion lower bound is that a good JSCC code must achieve distortions close to $D^*_{\infty}(\rho,P_S,\hat{Q}_{Y|X})$ simultaneously for all ``probable'' channel realizations. To show that this is impossible, we reduce the problem to that of broadcasting a source to two users with different channel conditions, corresponding to one empirical channel that is better than $Q_{Y|X}$, and one that is worse. An outer bound on the distortions, in the infinite-blocklength limit, was derived in~\cite{rfz06} for the quadratic-Gaussian case, and was recently generalized by Khezeli and Chen~\cite{kc15,kc16}. Here, we generalize these bounds for our scenario of interest, and show that for the binary Hamming case, with bandwidth expansion ratio $\rho>1$, it is impossible to design a JSCC code that will be optimal simultaneously for both channel conditions. Unfortunately, for the case $\rho<1$ this technique falls short of providing similar bounds, mainly because in this regime $D^*_{\infty}$ changes too slowly with the crossover probability of the BSC.

We note that it is still not clear whether the $n^{-1/2}$ scaling of $\Delta^*_n$ is achievable in general. In particular, a separation-based coding scheme can only achieve $\Delta_n=D_n-D^{*}_{\infty}= \Omega\left(\sqrt{\frac{\log n}{n}}\right)$, and to the best of our knowledge no JSCC coding scheme that achieves better scaling in general is known. See Section~\ref{sec:discuss}. Thus, despite the progress made in this work, the exact correct scaling of $\Delta^*_n$ remains an open question.

The structure of the paper is as follows. In Section~\ref{sec:results}, we define the binary instance of the JSCC problem we analyze, state the main result, and give a high-level sketch of proof. Section~\ref{sec:JSCCbounds} develops outer bounds on the distortions that can be achieved when sending a source over a broadcast channel. The bounds from Section~\ref{sec:JSCCbounds} are then specialized in Section~\ref{sec:sphericalBC}, for the problem of sending a binary source over a binary additive spherical noise, i.e., noise uniform on an Hamming sphere, broadcast channel. The full proof of our main result is given in Section~\ref{sec:proofmain}. Some auxiliary results for the problem of sending a source over a broadcast channel are given in Section~\ref{sec:aux}. The paper concludes with a discussion in Section~\ref{sec:discuss}.

\section{Main Result and Main Technical Contribution}
\label{sec:results}

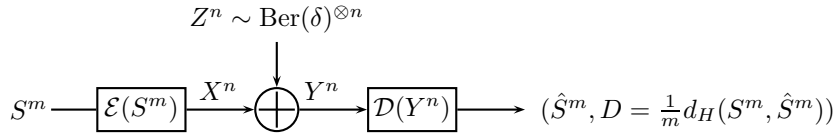
\begin{figure}[t]
	\begin{center}
		\psset{unit=0.6mm}
		\begin{pspicture}(0,50)(250,85)
			\rput(25,0){
				\rput(0,60){$S^m$}\psline(5,60)(15,60)\psframe(15,65)(35,55)\rput(25,60){$\m{E}(S^m)$}
				\psline{->}(35,60)(50,60)\rput(42,64){$X^n$}\pscircle(55,60){5}\psline(55,64)(55,56)\psline(51,60)(59,60)
				\psline{->}(55,75)(55,65)\rput(55,80){$Z^n\sim\Ber(\delta)^{\otimes n}$}
				\psline{->}(60,60)(75,60)\rput(65,64){$Y^n$}\psframe(75,65)(95,55)\rput(85,60){$\m{D}(Y^n)$}\psline{->}(95,60)(110,60)\rput(146,60){$(\hat{S}^m,D=\frac{1}{m}d_H(S^m,\hat{S}^m))$}
			}
		\end{pspicture}
	\end{center}
\caption{The $n(\rho,\delta)$ binary joint source-channel coding problem. it is assumed that $n=\rho m$.}
\label{fig:binaryJSCC}
\end{figure}

We study the binary symmetric joint source-channel coding problem, depicted in Figure~\ref{fig:binaryJSCC},
which is a special case of the general problem introduced above.  The source is binary symmetric $S\sim\Ber(1/2)$, the channel $Q_{Y|X}$ is $\mathrm{BSC}(\delta)$ with $\delta < 1/2$, the reconstruction alphabet is $\hat{\m{S}}=\{0,1\}$ and the distortion measure is Hamming, i.e., $d(S,\hat{S})=\Ind_{\{S\neq \hat{S}\}}$,  such that \[ d_H(S^m,\hat{S}^m)=\sum_{i=1}^m \Ind_{\{S_i\neq\hat{S}_i\}}. \] We use the binary entropy function\footnote{Throughout, logarithms are taken to the natural base.} \[ h_b(x)=-x\log x-(1-x)\log(1-x) \] and its inverse restricted to the interval $[0,1/2]$ as $h_b^{-1}(\cdot)$. It will be convenient to extend the domain of the function $h_b^{-1}(\cdot)$ to $(-\infty,\log 2]$, such that $h_b^{-1}(t)=0$ for all $t\leq 0$. For $0\leq a,b\leq 1$ we also define the binary convolution \[ a*b=a(1-b)+b(1-a). \] The expected distortion for this problem is formally defined below. For two binary variables, or vectors, the notation $+$ is to be understood as addition modulo-$2$.

\begin{definition} \label{def:binary_JSCC}
	Let $S^m\sim\Ber(1/2)^{\otimes m}$, $Y^n = X^n + Z^n$ with $Z^n\sim\Ber(\delta)^{\otimes n}$ independent of $X^n$, and $\rho=n/m$. The minimum expected Hamming distortion for transmitting $S^m$ over the channel from $X^n$ to $Y^n$ is defined as
	\begin{align}
	D^*(n,\rho,\delta)\triangleq \frac{1}{m}
	\min_{\substack{{\m{E}:\{0,1\}^m\to\{0,1\}^n}\\{\m{D}:\{0,1\}^n\to\{0,1\}^m}}} \mathbb{E}d_H(S^m,\m{D}(\m{E}(S^m)+Z^n)).
	\end{align}
	and its convergence rate function is
	\begin{align}
	\Delta^*_n=\Delta^*(n,\rho,\delta)\triangleq D^*(n,\rho,\delta)-D(\rho,\delta),
	\end{align}
	where
	\begin{align}
	D(\rho,\delta)&\triangleq h_b^{-1}(\log2-\rho(\log2-h_b(\delta))),\label{eq:DdeltaDef}
	\end{align}
	is the asymptotic value of $D^*(n,\rho,\delta)$.
\end{definition}

In this problem, it is well known that $\Delta^*_n = 0$ for all $n$, when $\rho=1$. Here, we will give a non-trivial lower bound for $\rho > 1$. We will express our result in terms of the following functions:
\begin{subequations}
\begin{align}
\Phi(\delta)&\triangleq \frac{2}{(1-2\delta)\log\left(\frac{1-\delta}{\delta}\right)}+\frac{1}{\delta(1-\delta)\log^2\left(\frac{1-\delta}{\delta}\right)}\label{eq:PhiDef}\\
f(\rho,\delta)&\triangleq \frac{1}{\rho}\frac{\Phi(\delta)}{\Phi(D(\rho,\delta))},\label{eq:fdef}\\
\eta(\rho,\delta)&\triangleq 2\rho\frac{\log\left(\frac{1-\delta}{\delta}\right)}{\log\left(\frac{1-D(\rho,\delta)}{D(\rho,\delta)}\right)}\cdot\frac{D(\rho,\delta)\left(1-f(\rho,\delta)\right)^2}{2f(\rho,\delta)+4D(\rho,\delta)\left(1-f(\rho,\delta)\right)}\cdot\frac{1+f(\rho,\delta)}{f(\rho,\delta)}\label{eq:etaDef}.
\end{align}
\end{subequations}
Our main result is the following. 
\begin{theorem}
	In the binary JSCC problem of Definition~\ref{def:binary_JSCC}, for all $\rho>1$ we have that
	\begin{align}
	\Delta^*_n \triangleq D^*(n,\rho,\delta) - D(\rho,\delta) \geq \sqrt \frac {\delta(1-\delta)} {2\pi n}  \eta(\rho,\delta) +\m{O}(n^{-3/4}\log{n}) , \label{eq:fdec}
	\end{align}
	where $\eta(\rho,\delta)$ is as defined in~\eqref{eq:etaDef} and is strictly positive for $D(\rho,\delta)>0$.
	\label{thm:main}
\end{theorem}
In particular, this binary symmetric example serves to show that there exists a choice of parameters $(\rho,P_S,Q_{Y|X})$, such that
\[  \Delta^*(n,\rho,P_S,Q_{Y|X})=\Omega\left(\frac{1}{\sqrt{n}}\right) . \]

The lower bound~\eqref{eq:fdec} is valid as long as $f(\rho,\delta)<1$, which is needed in order to justify~\eqref{eq:rhoconst} in the derivation below. In Lemma~\ref{lem:fdec} we show that for $\rho>1$ this is indeed the case, and it therefore suffices to require that $\rho>1$ in the statement of Theorem~\ref{thm:main}. Furthermore, $f(\rho,\delta)<1$ guarantees that $\eta(\rho,\delta)>0$ whenever $D(\rho,\delta)>0$, and consequently the bound is not trivial. In the regime $\rho\leq 1$, we have that $f(\rho,\delta)\geq 1$, and our bound is no longer valid.

\subsection{Outline of Reduction to JSCC Broadcast}
\label{subsec:pfoutline}

The proof of Theorem~\ref{thm:main} relies upon the reduction of the binary JSCC problem to the problem of sending a binary source over a broadcast channel, for which we then derive outer bounds on the achievable distortions region. We now outline this reduction. The details are straightforward but cumbersome, thus they are relegated to Section~\ref{sec:proofmain}. Here we use approximated equality or inequality, to say that the correction terms will be below the $1/\sqrt{n}$ order of interest. 

We restrict our attention to $\delta$ such that $\delta n$ is an integer; this reduction is insignificant in our scale of interest. We define the integer-valued random variable $K = w_H(Z^n) - \delta n$, where  $w_H(\cdot)$ is the Hamming weight of a vector. Let $\m{S}_{x,n}$ be the set of all length-$n$ binary sequences with Hamming weight $0\leq x\leq n$.
For a given encoder/decoder pair $(\m{E},\m{D})$ we define 
\begin{align}
\Psi(k) =\Psi_{\m{E},\m{D}}(k)&\triangleq \mathbb{E}\left[\frac{1}{m}d_H\left(S^m,\m{D}\left(\m{E}(S^m)+Z^n\right)\right) \ \bigg| \ K=k\right] \nonumber\\
&= \mathbb{E}\left[\frac{1}{m}d_H\left(S^m,\m{D}\left(\m{E}(S^m)+U^n\right)\right)\right] , \label{eq:PsiDef}
\end{align}
where $U^n\sim\Unif(\m{S}_{\delta n+k,n})$, although its dependence on $k$ is not made explicit.
We can then use iterated expectation to assert
\begin{align}
\mathbb{E}\left[\frac{1}{m}d_H\left(S^m,\m{D}\left(\m{E}(S^m)+Z^n\right)\right)\right] &= \mathbb{E}\left[\Psi(K)\right] \nonumber \\
&= \mathbb{E}\left[\mathbb{E}\left[\Psi(K) \ \bigg| \ |K|\right]\right].
\label{eq:iterexpt}
\end{align}
We then show that it is enough to consider a range of $K$ around $0$ that scales no faster than $\sqrt{n}$, and that within this range, since the Binomial distribution is approximately symmetric close to its mean, we have
\begin{align*}
\mathbb{E}\left[\Psi(K) \ \bigg| \ |K| = k_0 \right]&=\Pr(K=k_0 \ | \ |K|=k_0)\Psi(k_0)+ \Pr(K=-k_0 \ | \ |K|=k_0)\Psi(-k_0)\\
&\approx \frac{1}{2} \Psi(k_0) + \frac{1}{2} \Psi(-k_0).
\end{align*}
That is, we take $A=|K|/\sqrt{n}$ and consider bounded $A$. Thus,
\[ \mathbb{E}\left[\mathbb{E}\left[\Psi(K) \ \bigg| \ |K|\right]\right] \gtrapprox \frac{1}{2}\mathbb{E}\left[\mathbb{E}\left[\frac{1}{m}d_H\left(S^m,\m{D}\left(\m{E}(S^m)+U_{1}^n\right)\right) \right]+\mathbb{E}\left[ \frac{1}{m} d_H\left(S^m,\m{D}\left(\m{E}(S^m)+U_{2}^n\right)\right) \right]\right], \] where $U_{1}^n\sim\Unif\left(\m{S}_{n \left(\delta-\frac{A}{\sqrt{n}}\right),n}\right)$ and $U_{2}^n\sim\Unif\left(\m{S}_{n \left(\delta+\frac{A}{\sqrt{n}}\right),n}\right)$. The outer expectation in the right hand side is taken over $A$. 

\begin{figure}[]
	\begin{center}
		\psset{unit=0.6mm}
		\begin{pspicture}(0,0)(250,70)
		\rput(25,0){
			\rput(0,40){$S^m$}\psline{->}(5,40)(15,40)\psframe(15,45)(35,35)\rput(25,40){$\m{E}(S^m)$}
			\psline(35,40)(50,40)\rput(42,44){$X^n$}
			\psline{->}(50,40)(50,50)(70,50)\pscircle(75,50){5}
			\psline(75,54)(75,46)\psline(71,50)(79,50)
			\psline{->}(75,60)(75,55)\rput(75,65){$U^n_{1}$}
			\psline{->}(80,50)(95,50)\rput(87,54){$Y_1^n$}\psframe(95,55)(115,45)\rput(105,50){$\m{D}_1(Y_1^n)$}
			\psline{->}(115,50)(125,50)\rput(165,50){$\left(\hat{S}_1^m,D_1=\frac{1}{m}\mathbb{E}d_H(S^m,\hat{S}^m_1)\right)$}
			
			\psline{->}(50,40)(50,30)(70,30)\pscircle(75,30){5}\psline(75,34)(75,26)\psline(71,30)(79,30)
			\psline{->}(75,20)(75,25)\rput(75,15){$U^n_{2}$}
			\psline{->}(80,30)(95,30)\rput(87,34){$Y_2^n$}\psframe(95,35)(115,25)\rput(105,30){$\m{D}_2(Y_2^n)$}
			\psline{->}(115,30)(125,30)\rput(165,30){$\left(\hat{S}_2^m,D_2=\frac{1}{m}\mathbb{E}d_H(S^m,\hat{S}^m_2)\right)$}
		}
		\end{pspicture}
	\end{center}
\caption{Binary spherical noise JSCC broadcast problem.}
\label{fig:sphericalBC}
\end{figure}
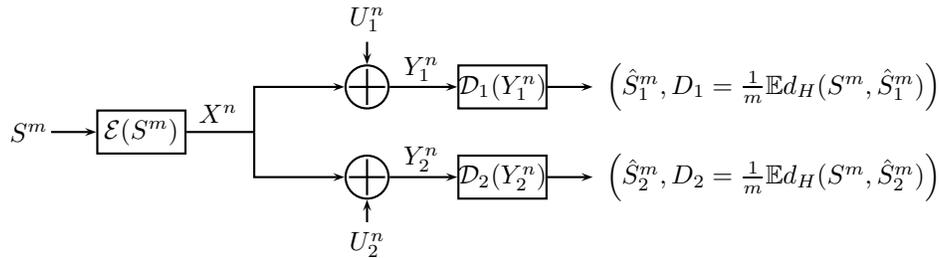

In the point-to-point JSCC setup considered, the encoder $\mathcal{E}$ and the decoder $\mathcal{D}$ are not aware of $A$, and cannot adapt to it. In particular, the same decoder is used whether $K$ is negative or positive. Nevertheless, we continue the analysis by making the following relaxations: we let the encoder and decoder vary as a function of $A$, and further for any value of $A$ we allow different decoders for the channel with additive noise $U_1^n$ (negative $K$) and the channel with additive noise $U_2^n$ (positive $K$). Thus we have that
\[ D^*(n,\rho,\delta) \gtrapprox \frac{1}{2}\mathbb{E} \left[ \underline{D}(A) \right], \]
where
\begin{align}
\underline{D}(a)\triangleq\frac{1}{m}\min\left(\mathbb{E}\left[d_H\left(S^m,\m{D}_1\left(\m{E}(S^m)+U_{1}^n\right)\right) \right]+\mathbb{E}\left[d_H\left(S^m,\m{D}_2\left(\m{E}(S^m)+U_{2}^n\right)\right) \right]\right),\label{eq:DaDef}
\end{align}
and the minimization is over all encoders $\m{E}:\{0,1\}^m\to\{0,1\}^n$ and decoders  $\m{D}_1:\{0,1\}^n\to\{0,1\}^m$, $\m{D}_2:\{0,1\}^n\to\{0,1\}^m$.
This last quantity is the sum-distortion in a problem of JSCC broadcast with spherical noise, see Figure~\ref{fig:sphericalBC}. We will next consider JSCC broadcast, and prove in Theorem~\ref{thm:Tradeoff}  that
\begin{align}
\underline{D}(a) \gtrapprox 2D(\rho,\delta)+\frac{a}{\sqrt{n}}\eta(\rho,\delta).
\label{eq:Dlb}
\end{align}
Thus, 
\[ D^*(n,\rho,\delta) \gtrapprox  D(\rho,\delta)+\frac{\mathbb E[A]}{2\sqrt{n}}\eta(\rho,\delta) . \]
The result of the theorem follows since $A$ is approximately the absolute value of a normal variable with zero mean and variance $\delta(1-\delta)$, thus
\[ \mathbb E[A] \approx \sqrt\frac{2\delta(1-\delta)}{\pi}, \]
which gives the stated result.

\subsection{Sending a Source Over a Broadcast Channel}
\label{sec:broadcast_outline}

We now outline the derivation of \eqref{eq:Dlb}, stating the steps which we will prove in Section~\ref{sec:sphericalBC}.
We consider the problem of sending a source over a broadcast channel, or simply, the \emph{JSCC broadcast} problem, as follows, See Figure~\ref{fig:JSCC_BS_gen_fig}.

\begin{figure}[]
	\begin{center}
		\psset{unit=0.6mm}
		\begin{pspicture}(0,50)(250,85)
		\rput(25,20){
			\rput(0,40){$S^m$}\psline(5,40)(15,40)\psframe(15,45)(35,35)\rput(25,40){$\m{E}(S^m)$}
			\psline(35,40)(50,40)\rput(42,44){$X^n$}\psframe(50,47)(80,33)\rput(65,40){$Q_{Y^n_1,Y^n_2|X^n}$}
			\psline(80,40)(90,40)
			\psline{->}(90,40)(90,50)(110,50)\rput(102,54){$Y_1^n$}\psframe(110,55)(130,45)\rput(120,50){$\m{D}_1(Y_1^n)$}
			\psline{->}(130,50)(140,50)\rput(155,50){$(\hat{S}_1^m,D_1)$}
			\psline{->}(90,40)(90,30)(110,30)\rput(102,34){$Y_2^n$}\psframe(110,35)(130,25)\rput(120,30){$\m{D}_2(Y_2^n)$}
			\psline{->}(130,30)(140,30)\rput(155,30){$(\hat{S}_2^m,D_2)$}
		}
		\end{pspicture}
	\end{center}
	\caption{JSCC over a Broadcast Channel.}
\label{fig:JSCC_BS_gen_fig}
\end{figure}
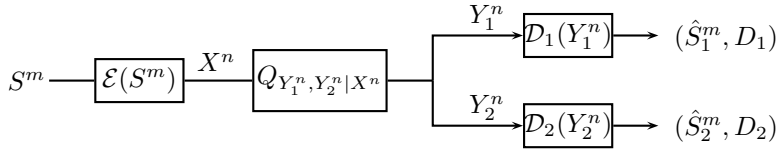

\begin{definition} \label{def:JSCC_BC_gen}
In the JSCC broadcast problem, an encoder observes a sequence $S^m=(S_1,\ldots,S_m)$ of $m$ i.i.d. samples generated according to the distribution $P_S$, and would like to send this sequence through the broadcast channel $Q_{Y^n_1,Y_2^n|X^n}$, which may not be memoryless, nor degraded. To that end, the encoder maps the source sequence $S^m$ to the channel input 
$X^n$ using an encoding function $\m{E}:\m{S}^m\to\m{X}^n$. The channel input $X^n$ is transmitted through the channel and the first receiver, which observes the channel output $Y_1^n$, generates an estimate $\hat{S}_1^m$ of the source sequence, using a decoding function $\m{D}_1:\m{Y}_1^n\to \hat{\m{S}}^m$, whereas the second receiver, which observes the channel output $Y_2^n$, generates an estimate $\hat{S}_2^m$ of the source sequence, using a decoding function $\m{D}_2:\m{Y}_2^n\to \hat{\m{S}}^m$. We assume that the reconstruction alphabets $\hat{\m{S}}_1,\hat{\m{S}_2}$ are identical, and that the quality of the two reconstructions are measured with respect to the same separable distortion measure $d:\m{S}\times\hat{\m{S}}\to \RR$:\footnote{These assumptions are made for the sake of simplicity only, results easily go through to the more general case as well.}  
\begin{align}
D_i=\frac{1}{m}\mathbb{E}d(S^m,\hat{S}_i^m), \ i=1,2,
\end{align}    
\end{definition}

\vspace{2mm}

For given $m$, $n$, $P_S$, and $Q_{Y^n_1,Y_2^n|X^n}$, a distortion pair $(D_1,D_2)$ is said to be achievable if there exist $(\m{E},\m{D}_1,\m{D}_2)$ such that $\frac{1}{m}\mathbb{E}[d(S^m,\hat S_i^m)] \leq D_i$ for $i=1,2$. It will be convenient to express results using the \emph{bandwidth expansion factor} $\rho=n/m$. Our goal is to establish an outer bound on the achievable pairs $(D_1,D_2)$.

For our results, we need the following functions of the source. We define an auxiliary variable $U$ via a conditional distribution $P_{U|S}$. By combining with the given $P_S$ we obtain $P=P_{SU}$. With respect to this distribution, we define:
\begin{align} 
F_P(t)&\triangleq\min_{\substack{{V \ : \ U-S-V} \\ {I(S;V)\geq t} } } I(S;V|U)  , \label{eq:FpDef}\\
\bar{R}_P(D)&\triangleq \min_{\substack{{\hat{S} \ : \ U-S-\hat{S}} \\ {\mathbb{E}d(S,\hat{S})\leq D}}} I(U;\hat{S}).\label{eq:RpDef}
\end{align}
Notice that when $U=S$ the function $\bar{R}_P(D)$ reduces to the rate-distortion function $R_{P_S}(D)$ of the source $S$.\footnote{We note that the function $\bar{R}_P(D)$ also arises as an upper bound on the communication rate required in order to perform a certain coordination task, see~\cite{sr18}.} Furthermore, we define the following function of the channel 
$Q^n=Q_{Y^n_1,Y^n_2|X^n}(y_1^n,y_2^n|x^n)$, 
\begin{align} 
G_{Q^n}(t)& \triangleq\max_{\substack{ {W,X^n \ : \ W-X^n-(Y^n_1,Y^n_2)} \\ { I(X^n;Y^n_1|W)\geq t}      } }
I(Y^n_2;W) .
\label{eq:GqDef} \end{align}
Note the relation to the capacity region of the broadcast channel: If $nR_1$ and $nR_2$ bits can be communicated reliably to the receivers $Y_1^n$ and $Y_2^n$, respectively, then $n R_2\leq G_{Q^n}(nR_1)$~\cite[Chapter 5.4.1]{ElGamalKim} .

Let $R_{P_S}(D)$ be the rate-distortion function of the source $P_S$. We prove the following theorem.

\begin{theorem} \label{thm:JSCCgen}
	Consider the problem of transmitting $m$ realizations of the i.i.d. source $S\sim P_S$, over the $n$-letter broadcast channel $Q^n$. If $(D_1,D_2)$ is achievable, then for any $P_{U|S}$ we have that
	\begin{align}
	\bar R_P(D_2) \leq \frac{1}{m} \cdot G_{Q^n}\left(m F_P\left(R_{P_S}(D_1) \right) \right),\label{eq:JSCCbound}
	\end{align}
	where $P=P_S P_{U|S}$ is the joint distribution on $(S,U)$ induced by the choice of $P_{U|S}$.
\end{theorem}

The proof of this bound, which is based on generalizing a technique developed by Reznic, Feder and Zamir~\cite{rfz06} for the Gaussian joint source-channel coding broadcast problem, will be given in Section~\ref{sec:JSCCbounds}. 
While the functions $F_P(t)$, $\bar{R}_P(D)$ and $R_{P_S}(D)$, involved in the evaluation of the bound from Theorem~\ref{thm:JSCCgen} only require solving a single-letter optimization, the function $G_{Q^n}$ requires solving, or bounding the solution of, an $n$-letter optimization, which is in general a challenging task. However, if the broadcast channel is memoryless and degraded then $G_{Q^n}$ single-letterizes; as a consequence the whole bound becomes single-letter as well, and in fact we will also show that for this case it can be obtained as a special case of~\cite[Theorem 5]{kc15}. 

We now speicialize Theorem~\ref{thm:JSCCgen} to the \emph{spherical-noise binary JSCC broadcast problem}. Recalling Figure~\ref{fig:sphericalBC}, it is defined as follows.
\begin{definition} \label{def:JSCC_BC_binary}
The spherical-noise binary JSCC broadcast problem is a JSCC broadcast problem (as in Definition~\ref{def:JSCC_BC_gen}) with
 a memoryless symmetric binary source $S^m$, channels $Y_1^n=X^n+U_1^n$, $Y_2^n=X^n+U_2^n$, where $U_1^n\sim\Unif\left(\m{S}_{n\delta_1,n}\right)$ and $U_2^n\sim\Unif\left(\m{S}_{n(\delta_1*\delta_2),n}\right)$, $(U_1^n,U_2^n)\indep X^n$, and Hamming distortion. \end{definition}

In this problem the broadcast channel 
is neither memoryless, nor degraded. 
Indeed, if we replaced the spherical noises by an i.i.d. noise with the same marginals, the channel would become a $\mathrm{BSC}(\delta_1,\delta_1*\delta_2)$ degraded memoryless broadcast channel. In that case, the corresponding channel function does single-letterize, and equals 
\begin{align}  \frac{1}{n} G_{Q^n}(nt) &= G_\mathrm{BSC}(t) \nonumber \\ &= \log 2-h_b\left(\delta_2*h_b^{-1}\left(h_b(\delta_1)+t\right)\right),\end{align} see Section~\ref{sec:aux}. The i.i.d. JSCC bound derived from this function would easily prove our result. Closing the gap between the i.i.d. case and the spherical-noise case, is one of the main technical challenges in this work.

For the spherical-noise broadcast channel, we prove that 
\[ \frac{1}{n} G_{Q^n}(nt)\leq G_\mathrm{BSC}(t) + \Gamma(n,\delta_2),
\]
where
\[
\Gamma(n,\delta_2)\triangleq\sqrt{\frac{\delta_2}{n}}\log\left(\frac{n}{\delta_2} \right)+\frac{\log{n}+1}{2n}.\]

The proof of this bound, which is given in Section~\ref{sec:sphericalBC}, is based on replacing $U_2^n$ with the noise $U_1^n+\tilde{Z}^n$, where $\tilde{Z}^n\sim\Ber(\delta_2)^{\otimes n}$, such that the obtained broadcast channel $\tilde{Q}^n$ is degraded (though not memoryless), and its corresponding function $G_{\tilde{Q}^n}$ can be computed using Mrs. Gerber's Lemma~\cite{wz73} (see below). The difference $|G_{\tilde{Q}^n}(nt)-G_{Q^n}(nt)|$ is essentially bounded by $\sup_{P_{X^n}}|H(X^n+U_1^n+\tilde{Z}^n)-H(X^n+U_2^n)|$, which can be bounded via a coupling argument introduced in~\cite{pw16} (see also~\cite{op18}). Using this technique, we show that $|G_{\tilde{Q}^n}(nt)-G_{Q^n}(nt)| \lessapprox n\cdot\Gamma(n,\delta_2)$. 

Thus, in order to apply Theorem~\ref{thm:JSCCgen} for obtaining an outer bound on the achievable $(D_1,D_2)$ pairs for the binary spherical noise JSCC broadcast problem, we need only choose an auxiliary channel $P_{U|S}$ and evaluate the functions $F_P(t)$, $\bar{R}_P(D)$ and $R_{P_S}(D)$. For the particular choice $U=S+N$, where $N\sim\Ber(q)$ is statistically independent of $X$, evaluating these functions becomes particularly simple, and we obtain the following Theorem, proved in Section~\ref{sec:sphericalBC}.

\begin{theorem}
	Consider the binary spherical noise JSCC broadcast problem. If $(D_1,D_2)$ is achievable, then for any $0<q<1/2$
	\begin{align}
	\log 2-h_b(q*D_2) &\leq \rho\left[\log 2-h_b\left(\delta_2*h_b^{-1}\left(h_b(\delta_1)+\frac{h_b(q*D_1)-h_b(D_1)}{\rho}\right)\right)\right]+\rho\Gamma(n,\delta_2),\label{eq:basicbound}
	\end{align}
	where $\Gamma(n,\delta_2)$ is as defined in~\eqref{eq:epsilonndef}.
	\label{thm:JSCC_Delta}
\end{theorem}

Recall now that our goal is to bound from $\underline{D}(a)$ \eqref{eq:DaDef}. This is nothing but the sum $D_1+D_2$ in the binary spherical noise JSCC broadcast problem, with appropriate $\delta_1$ and $\delta_2$. The required result is as follows. 

\begin{theorem} 
	Consider the binary spherical noise JSCC broadcast problem, with $\delta_1=\delta-\tfrac{a}{\sqrt{n}}$, $\delta_1*\delta_2=\delta+\tfrac{a}{\sqrt{n}}$, for some positive $a<\log^2(n)$. If $(D_1,D_2)$ is achievable, then
	\begin{align}
	D_1+D_2\geq 2D(\rho,\delta)+ \frac{a\eta(\rho,\delta)}{\sqrt{n}}+\m{O}\left(n^{-3/4}\log{n}\right),
	\end{align}
	where $\eta(\rho,\delta)$ is as defined in~\eqref{eq:etaDef}.
	\label{thm:Tradeoff}
\end{theorem}

The proof appears in Section~\ref{subsec:from4to7}. Given Theorem~\ref{thm:JSCC_Delta}, it is merely a matter of algebraic manipulations and local approximations. Note that since we are interested in the regime where $\delta_2=\Theta(n^{-1/2})$, the term $\Gamma(n,\delta_2)$ only contributes to the $\m{O}(n^{-3/4}\log{n})$ correction.

\vspace{2mm}

As we make extensive use of Mrs. Gerber's Lemma (MGL)~\cite{wz73}, for completeness, we end this section with its statement.

\begin{lemma}[Mrs. Gerber's Lemma~\cite{wz73}]
The MGL function $\varphi_{\delta}(t)=h_b(\delta*h_b^{-1}(t))$ is convex in $t$. Furthermore, for $Z^n\sim\Ber(\delta)^{\otimes n}$ and any $n$-dimensional binary vector $X^n$, statistically independent of $Z^n$, we have that
\begin{align*}
H(X^n+Z^n)\geq n \varphi_{\delta}\left(\frac{H(X^n)}{n} \right).
\end{align*}
\end{lemma}

\section{Outer Bound for The General JSCC Broadcast Problem}
\label{sec:JSCCbounds}

In this section we address the general JSCC broadcast problem of Definition~\ref{def:JSCC_BC_gen}. We prove Theorem~\ref{thm:JSCCgen} by analyzing the source and channel functions  \eqref{eq:FpDef}-\eqref{eq:GqDef}. We note that the setting of Theorem~\ref{thm:JSCCgen} is asymmetric: the source is assumed to be i.i.d., while the channel is neither memoryless nor degraded; this is the most general form that we need for this work, and the extension for sources with memory will become quite obvious in the sequel.

\subsection{The Source Functions}

First, we prove two simple statements regarding properties of the functions $F_P$ \eqref{eq:FpDef} and $\bar{R}_P$ \eqref{eq:RpDef}. Recall that for a given choice of auxiliary conditional distribution $P_{U|S}$, we have defined $P=P_S P_{U|S}$. The source function $F_{P^{\otimes m}}$ corresponding to $m$ i.i.d. draws $(U^m,S^m)$ from $P$ is
\begin{align*}    
F_{P^{\otimes m}}(t)&\triangleq\min_{\substack{{V \ : \ U^m-S^m-V} \\ {I(S^m;V)\geq t} } }I(S^m;V|U^m).
\end{align*}
 
\begin{lemma} \label{lem:F_P}
	The function $F_P(t)$ is monotone non-decreasing and convex. Furthermore, it tensorizes, i.e.,
	\begin{align*}
	F_{P^{\otimes m}}(mt)  =mF_P (t).
	\end{align*}
\end{lemma}

\begin{proof}
	Monotonicity of the function $F_P(t)$ follows by definition. 
	For convexity, let $V_0$ and $V_1$ be random variables (induced by the channels $P_{V_i|S}$) that attain $F_P(t_0)$ and $F_P(t_1)$ with equality. Let $A\sim\Ber(\alpha)$ be statistically independent of $(U,S)$, and define $\tilde{V}=(A,V_A)$. We have that
	\begin{align}
	I(S;\tilde{V})&=I(S;V_A|A)\nonumber\\
	&=(1-\alpha) I(S;V_0)+\alpha I(S;V_1)\nonumber\\
	&=(1-\alpha) t_0+\alpha t_1\nonumber,
	\end{align}
	and by definition of $F_p(t)$,
	\begin{align}
	F_p((1-\alpha) t_0+\alpha t_1)&\leq I(S;\tilde{V}|U)\nonumber\\
	&=(1-\alpha) I(S;V_0|U)+\alpha I(S;V_1|U)\nonumber\\
	&=(1-\alpha) F_{P}(t_0)+\alpha F_P(t_1).
	\end{align}
	We prove tensorization by induction. Let $(U^m,S^m)\sim P^{\otimes m}$. For any $V$ that satisfies the Markov chain $U^m-S^m-V$, we have
	\begin{align}
	&F_P\left(\frac{I(S^m;V)}{m}\right)=F_P\left(\frac{I(S^{m-1};V)+I(S_m;V|S^{m-1})}{m}\right)\nonumber\\
	&=F_P\left(\frac{I(S^{m-1};V)+I(S_m;V,S^{m-1})}{m}\right)\nonumber\\
	&=F_P\left(\frac{m-1}{m}\frac{I(S^{m-1};V)}{m-1}+\frac{1}{m}I(S_m;V,S^{m-1})\right)\nonumber\\
	&\leq \frac{m-1}{m}F_P\left(\frac{I(S^{m-1};V)}{m-1}\right)+\frac{1}{m}F_P\left(I(S_m;V,S^{m-1})\right)\nonumber,
	\end{align}
	where we have used the convexity of $t\mapsto F_P(t)$ in the last inequality. Invoking the induction hypothesis, we have
	\begin{align}
	&F_P\left(\frac{I(S^m;V)}{m}\right)\nonumber\\
	&\leq \frac{1}{m}F_{P^{\otimes(m-1)}}\left(I(S^{m-1};V)\right)+\frac{1}{m}F_P\left(I(S_m;V,S^{m-1})\right)\nonumber\\
	&\leq \frac{1}{m}\left[I(S^{m-1};V|U^{m-1})+I(S_m;V,S^{m-1}|U_m)\right], \label{eq:F2bound}
	\end{align}
	where the last inequality follows by definition of $F_{P^{\otimes(m-1)}}$ and $F_P$ and the fact $U^{m-1}-S^{m-1}-V$ and $U_m-S_m-(V,S^{m-1})$ are indeed Markov chains. Noting that
	\begin{align}
	I(S^{m-1};V|U^{m-1})&\leq I(S^{m-1};V|U^{m}),\nonumber
	\end{align}
	and
	\begin{align}
	I(S_{m};V,S^{m-1}|U_m)&\leq I(S_m;V,S^{m-1}|U^{m})\nonumber\\
	&= I(S_m;V|S^{m-1},U^{m}),\nonumber
	\end{align}
	which both follow since $S^m$ is memoryless, we obtain
	\begin{align}
	I(S^{m-1};V|U^{m-1})+I(S_m;V,S^{m-1}|U_m)\leq I(S^m;V|U^m).\label{eq:MIchain}
	\end{align}
	Substituting~\eqref{eq:MIchain} into~\eqref{eq:F2bound}, gives
	\begin{align}
	I(S^m;V|U^m)\geq m F_P\left(\frac{I(S^m;V)}{m}\right).
	\end{align}
	Thus, we have shown that $F_{P^{\otimes m}}(mt)\geq m F_P(t)$. On the other hand, we have that
	\begin{align}
	    F_{P^{\otimes m}}(mt)&\triangleq\min_{\substack{{V \ : \ U^m-S^m-V} \\ {I(S^m;V)\geq mt} } }I(S^m;V|U^m)\nonumber\\
	    &\leq \min_{\substack{{V^m \ : \ U^m-S^m-V^m} \\ {I(S^m;V^m)\geq mt} } }I(S^m;V^m|U^m)\label{eq:miniid}\\
	    &=m F_P(t),\nonumber
	\end{align}
	where the minimum in~\eqref{eq:miniid} is taken over the memoryless channels $P_{V^m|S^m}=\prod_{i=1}^n P_{V|S}$.
\end{proof}

\begin{lemma}\label{lem:Rbar}
	Let $(U^m,S^m)\sim P^{\otimes m}$, and let $\hat{S}^m$ be a random vector satisfying the Markov chain $U^m-S^m-\hat{S}^m$ and $\frac{1}{m}\mathbb{E} d(S^m, \hat S^m) \leq D$, then
	\begin{align}
	I(U^m;\hat{S}^m)\geq m\bar{R}_P(D).\nonumber
	\end{align}
\end{lemma}

\begin{proof}
	Since $U^m$ is memoryless, we have that
	\begin{align}
	I(U^m;\hat{S}^m)\geq \sum_{i=1}^m I(U_i;\hat{S}_i).
	\end{align}
	Note that $\frac{1}{m}\sum_{i=1}^m \mathbb{E}d(S_i;\hat{S}_i)\leq D$ by separability of $d(S^m;\hat{S}^m)$, and that the Markov chain $U^m-S^m -\hat{S}^m$ implies that $U_i-S_i-\hat{S}_i$ is also a Markov chain. It is easy to see that the function $D\mapsto\bar{R}_P(D)$ is convex. Thus, letting $d_i=\mathbb{E}d(S_i;\hat{S}_i)$, we have that
	\begin{align}
	I(U^m;\hat{S}^m)\geq m\frac{1}{m}\sum_{i=1}^m \bar{R}_P(d_i)\geq m \bar{R}_P(D).
	\end{align}
\end{proof}

\subsection{The Channel Function and Derivation of Theorem~\ref{thm:JSCCgen}} 

In the general (non-memoryless, non-degraded) case the channel function $G_{Q^n}$ \eqref{eq:GqDef} does not tenzorise. However, we have the following basic properties.

\begin{lemma} \label{lem:no_tensor}
	The function $G_{Q^n}(t)$ is monotone non-increasing and concave. 
\end{lemma}

\begin{proof}
	Monotonicity of $G_{Q^n}(t)$ follows by definition. 
	For concavity, let $(W_0,X^n_0)$ and $(W_1,X^n_1)$ be random variables that attain $G_{Q^n}(t_0)$ and $G_{Q^n}(t_1)$ with equality. Let $A\sim\Ber(\alpha)$ be statistically independent of $(W_0,X^n_0,W_1,X^n_1)$, and define $(\tilde{W},\tilde{X}^n)=((A,W_A),X^n_A)$. We have that
	\begin{align}
	I(\tilde{X}^n;Y^n_1|\tilde{W})&=I(X^n_A;Y^n_1|W_A,A)\nonumber\\
	&=(1-\alpha) I(X^n_0;Y^n_1|W_0)+\alpha I(X^n_1;Y^n_1|W_1)\nonumber\\
	&= (1-\alpha)t_0+\alpha t_1,
	\end{align}
	and by definition of $G_{Q^n}(t)$,
	\begin{align}
	G_{Q^n}(\alpha t_0+(1-\alpha)t_1)&= G_{Q^n}(I(\tilde{X}^n;Y^n_1|\tilde{W}))\nonumber\\
	& \geq I(Y^n_2;\tilde{W})\nonumber\\
	&=I(Y^n_2;A)+I(Y^n_2;W_A|A)\nonumber\\
	&\geq \alpha G_{Q^n}(t_0)+(1-\alpha)G_{Q^n}(t_1).
	\end{align}
\end{proof}

We are now in a position to prove Theorem~\ref{thm:JSCCgen}. The proof is essentially a generalization of the technique developed by Reznik, Feder and Zamir for the Gaussian joint source-channel coding problem~\cite{rfz06}. Their proof relied heavily on the entropy-power inequality (EPI), which is replaced by the functions $F_{P}(t)$ and $G_{Q^n}(t)$ in the proof below. We remark that although Theorem~\ref{thm:JSCCgen} is stated and proved for channels without a cost constraint, such a constraint can be included by constraining the distribution of $X^n$ in the computation of $G_{Q^n}(t)$, in the obvious way.

\begin{proof}[Proof of Theorem~\ref{thm:JSCCgen}]
	Let $\hat{S}_1^m,\hat{S}_2^m$ be the estimates produced from the outputs $Y_1^n$ and $Y_2^n$, respectively. We have
	\begin{align}
	m\bar{R}_P(D_2)&\leq I(U^m,\hat{S}_2^m)\label{eq:RpDapp}\\
	&\leq I(U^m;Y_2^n)\label{eq:BC1}\\
	&\leq G_{Q^n}(I(X^n;Y_1^n|U^m))\label{eq:BC2}\\
	&\leq G_{Q^n}\left(I(S^m;Y_1^n|U^m) \right)\label{eq:BC4}\\
	&\leq G_{Q^n}\left(F_{P^{\otimes m}}\left(I(S^m;Y_1^n)\right) \right)\label{eq:BC5}\\
	&= G_{Q^n}\left(m F_{P}\left(\frac{I(S^m;Y_1^n)}{m}\right) \right)\label{eq:BC6}\\
	&\leq G_{Q^n}\left(m F_{P}\left(\frac{I(S^m;\hat{S}_1^m)}{m}\right) \right)\label{eq:BC7}\\
	&\leq G_{Q^n}\left(m F_{P}\left(R(D_1)\right) \right), \nonumber
	\end{align}
	where~\eqref{eq:RpDapp} follows from Lemma~\ref{lem:Rbar}, ~\eqref{eq:BC1} follows from the data processing inequality (DPI),~\eqref{eq:BC2} from definition of $G_{Q^n}$,~\eqref{eq:BC4} from the DPI and monotonicity of $G_{Q^n}$,~\eqref{eq:BC5} from definition of $F_{P^{\otimes m}}$,~\eqref{eq:BC6} from tensorization of $F_{P^{\otimes m}}$, and~\eqref{eq:BC7} from the DPI.
\end{proof}

Note that $U^m$ plays a two-fold role here: in~\eqref{eq:BC2} we used the Markov chain $U^m-X^n-(Y^n_1,Y^n_2)$, whereas in~\eqref{eq:RpDapp} and~\eqref{eq:BC5} we used $U^m-S^m-Y^n_1$. Thus, the source functions $F_P(t)$ and $\bar{R}_P(D)$, and the broadcast function $G_{Q^n}(t)$ are \emph{coupled} via the same auxiliary variable $U^m$. This is also the main weakness of the bound above: Even though the same $U^m$, whose distribution is fixed and memoryless once we choose the channel $P_{U|S}$, appears in both Markov chains, in the transition from~\eqref{eq:BC1} to~\eqref{eq:BC2}, we have used the definition of $G_{Q^n}$, which involves a \emph{maximization} with respect to $U^m$. As will be shown in the sequel, in the special case where $Q^n$ is degraded and memoryless, the auxiliary random variables achieving the maximum in the definition of $G_{Q^n}(t)$ are of the form $(W^n,X^n)\sim P_{WX}^{\otimes n}$, i.e., $n$-letter memoryless distribution. We therefore see that for such $Q^n$, the random variables $(U^m,X^n)$, where $U^m$ is $m$-letter memoryless, cannot achieve the maximum in the definition of $G_{Q^n}(t)$, unless $m=n$. Thus, the inequality~\eqref{eq:BC2} must be strict in this case.

\section{Outer Bound for the Binary Spherical-Noise JSCC Broadcast Problem}
\label{sec:sphericalBC}

In this section we derive an explicit bound for the binary spherical-noise JSCC broadcast problem (Definition~\ref{def:JSCC_BC_binary}), namely we prove Theorem~\ref{thm:Tradeoff}. To that end, we first establish Theorem~\ref{thm:JSCC_Delta} .

\subsection{Proof of Theorem~\ref{thm:JSCC_Delta}}

For specializing Theorem~\ref{thm:JSCCgen} to the binary case, we need to evaluate or at least bound the source and channel functions. The crucial part is a bound on the (non single-letter) channel function, as follows.

\begin{lemma} 
Let $Q^n=Q_{Y^n_1,Y^n_2|X^n}(y_1^n,y_2^n|x^n)$ be the additive spherical-noise broadcast channel $Y_1^n=X^n+U_1^n$, $Y_2=X^n+U_2^n$, where $U_1^n\sim\Unif\left(\m{S}_{n\delta_1,n}\right)$ and $U_2^n\sim\Unif\left(\m{S}_{n(\delta_1*\delta_2),n}\right)$, $(U_1^n,U_2^n)\indep X^n$. Then,
\begin{align}
\frac{1}{n} G_{Q^n}(nt)\leq \log 2-h_b\left(\delta_2*h_b^{-1}\left(h_b(\delta_1)+t\right) \right) + \Gamma(n,\delta_2),
\end{align}
where
\begin{align}
\Gamma(n,\delta_2)\triangleq\sqrt{\frac{\delta_2}{n}}\log\left(\frac{n}{\delta_2} \right)+\frac{\log{n}+1}{2n}.\label{eq:epsilonndef}
\end{align}
\label{lem:GsphericalBC}
\end{lemma}

\begin{proof}
Let $(W,X^n)$ satisfy the Markov chain $W-X^n-(Y_1^n=X^n+U_1^n,Y_2^n=X^n+U_2^n)$. We begin by writing
\begin{align}
H(Y_2^n|W)=H(X^n+U_1^n+Z_3^n|W)+\left[H(X^n+U_2^n|W)-H(X^n+U_1^n+Z_3^n|W) \right],\label{eq:condEnt}
\end{align}
where $Z_3^n\sim\Ber(\delta_2)^{\otimes n}$. We will upper bound the absolute value of the term in the square brackets via coupling. Consider the following joint distribution on $(U_2^n,U_1^n+Z_3^n)$:
\begin{itemize}
\item Let $\Pi$ be a uniform random permutation on $[n]=\{1,\ldots,n\}$.
\item Let $T=T_0+T_1$ where $T_0\sim\mathrm{Binomial}(n(1-\delta_1),\delta_2)$ and $T_1\sim\mathrm{Binomial}(n\delta_1,1-\delta_2)$ are independent.
\item Set $U_{2,\Pi(i)}=1$ for $i=1,\ldots,(\delta_1*\delta_2) n$ and $U_{2,\Pi(i)}=0$ for $i=(\delta_1*\delta_2) n+1,\ldots,n$.
\item Set $U_{1,\Pi(i)}+Z_{3,\Pi(i)}=1$ for $i=1,\ldots,T$ and $U_{1,\Pi(i)}+Z_{3,\Pi(i)}=0$ for $i=T+1,\ldots,n$.
\end{itemize}
Clearly $U_2^n$ and $U_1^n+Z_3^n$ have the correct marginal distributions. Moreover, the expected Hamming distance between these vectors satisfies
\begin{align}
\mathbb{E}w_H(U_2^n+(U_1^n+Z_3^n))&=\mathbb{E}|T-n(\delta_1*\delta_2)|\nonumber\\
&=\mathbb{E}\sqrt{\left(T-n(\delta_1*\delta_2) \right)^2}\nonumber\\
&\leq \sqrt{\Var(T)}\nonumber\\
&=\sqrt{\Var(T_1)+\Var(T_2)}\nonumber\\
&=\sqrt{n\delta_2(1-\delta_2)}\nonumber\\
&\leq\sqrt{n\delta_2},
\end{align}
where the first inequality follows from Jensen's inequality and the fact that $\mathbb{E}(T)=n(\delta_1*\delta_2)$. Now, applying~\cite[Proposition 8]{pw16} (see also~\cite{op18}), we obtain for $\delta_2<1/2$
\begin{align}
\left|H(X^n+U_2^n|W)-H(X^n+U_1^n+Z_3^n|W) \right|\leq \sqrt{n\delta_2}\log\left(\frac{n}{\delta_2} \right).
\end{align}
Thus, we can use Mrs. Gerber's Lemma (MGL) to lower bound~\eqref{eq:condEnt} as
\begin{align}
H(Y_2^n|W) & \geq H(X^n+U_1^n+Z_3^n|W)-\sqrt{n\delta_2}\log\left(\frac{n}{\delta_2} \right)\nonumber\\
&\geq n \varphi_{\delta_2} \left(\frac{H(X^n+U_1^n|W)}{n} \right) -\sqrt{n\delta_2}\log\left(\frac{n}{\delta_2} \right), \label{eq:mgl}
\end{align}
where $\varphi_{\delta_2}(x)=h_b(\delta_2*h^{-1}(x))$ is the MGL function. Further bounding, we have
\begin{align}
&H(Y_2^n|W) = n \varphi_{\delta_2} \left(\frac{H(X^n+U_1^n|X^n)+H(X^n+U_1^n|W)-H(X^n+U_1^n|X^n,W)}{n} \right) -\sqrt{n\delta_2}\log\left(\frac{n}{\delta_2} \right)\nonumber\\
&= n \varphi_{\delta_2} \left(\frac{H(U_1^n)+I(X^n;Y_1^n|W)}{n} \right) -\sqrt{n\delta_2}\log\left(\frac{n}{\delta_2} \right)\nonumber\\
&\geq n \varphi_{\delta_2} \left(\frac{n h_b(\delta_1)-\frac{1}{2}(\log{n}+1)+I(X^n;Y_1^n|W)}{n} \right)-\sqrt{n\delta_2}\log\left(\frac{n}{\delta_2} \right)\label{eq:spherent}\\
&= n \varphi_{\delta_2} \left(h_b(\delta_1)+\frac{I(X^n;Y_1^n|W)}{n}-\frac{\log{n}+1}{2n} \right)-\sqrt{n\delta_2}\log\left(\frac{n}{\delta_2} \right)\nonumber\\
&\geq n \varphi_{\delta_2} \left(h_b(\delta_1)+\frac{I(X^n;Y_1^n|W)}{n} \right)-\frac{\log{n}+1}{2}\varphi_{\delta_2}'\left(h_b(\delta_1)+\frac{I(X^n;Y_1^n|W)}{n} \right)-\sqrt{n\delta_2}\log\left(\frac{n}{\delta_2} \right),\nonumber
\end{align}
where in~\eqref{eq:spherent} we substitute~\cite[Chapter 10, Lemma 7]{ms77} to lower-bound $H(U_1^n)$, and in the last inequality we have defined  the MGL derivative $\varphi_{\delta_2}'(x)=\frac{d}{dx}\varphi_{\delta_2}(x)$, and used the convexity of $x\mapsto\varphi_{\delta_2}(x)$~\cite{wz73}. Recalling that $\varphi_{\delta_2}'(x)\leq 1$ due to~\cite[Theorem 2.6]{ww75} (or alternatively, as can be seen directly from the expression for $\varphi_{\delta_2}'(x)$ derived in~\cite{wz73}), we have obtained
\begin{align}
H(Y_2^n|W)\geq n h_b\left(\delta_2*h_b^{-1}\left(h_b(\delta_1)+\frac{I(X^n;Y_1^n|W)}{n} \right) \right)-n\Gamma(n,\delta_2),
\end{align}
and the claim now follows since $H(Y_2^n)\leq n\log 2$.
\end{proof}

The treatment of the source functions is much simpler. As it is identical to the i.i.d. binary problem, we defer the derivation to  Section~\ref{subsec:binary}, where we show that for $S\sim\Ber(p)$ and the choice $U=S+N$, where $N\sim\Ber(q)$ is statistically independent of $S$, we have:
\begin{align}
 	F_P(t) & \geq t-h_b(q*p)+h_b\left(q*h_b^{-1}\left(h_b(p)-t\right) \right)\\
	\bar{R}_P(D)&=h_b(q*p)-h_b(q*D),
   \end{align}
for $0 \leq D \leq p$. Substituting these expressions for $p=1/2$ and the channel-function bound of Lemma~\ref{lem:GsphericalBC} in Theorem~\ref{thm:JSCCgen}, Theorem~\ref{thm:JSCC_Delta} is immediately obtained.

\subsection{Local Analysis: From Theorem~\ref{thm:JSCC_Delta} to Theorem~\ref{thm:Tradeoff}}
\label{subsec:from4to7}

Our next goal is to manipulate the bound from Theorem~\ref{thm:JSCC_Delta} in order to obtain a lower bound on the sum-distortion. 
The proof of the following Lemma is brought in Appendix~\ref{app:D1D2lemma}, and is based on several auxiliary lemmas, which are stated and proved in Appendix~\ref{app:aux}.

\begin{lemma}
Consider the binary spherical-noise JSCC broadcast problem of Definition~\ref{def:JSCC_BC_binary}. If $(D_1,D_2)$ is achievable, then for any $\tau>0$ we have
\begin{align}
D_2-D_1&\leq\frac{1+2D_2\tau}{2D_2\tau} \frac{\rho C(\delta_1*\delta_2)-R(D_2)+\rho\Gamma(n,\delta_2)}{\log\left(\frac{1-D_2}{D_2}\right)}+(\delta_1*\delta_2-\delta_1)\frac{\log\left(\frac{1-\delta_1}{\delta_1}\right)}{\log\left(\frac{1-D_2}{D_2}\right)}\frac{\Phi(\delta_1)}{\Phi(D_2)}(1+\tau)\frac{g(D_1)}{g(D_2)},\label{eq:GapIneq3}
\end{align}
where $\Phi(t)$ is defined in~\eqref{eq:PhiDef}, $\Gamma(n,\delta_2)$ is defined in~\eqref{eq:epsilonndef}, and
\begin{align}
C(t)=R(t)&\triangleq\log{2}-h_b(t),\label{eq:RdDef}\\
g(t)&\triangleq(1-2t)\log\left(\frac{1-t}{t}\right)=(1-2t)h_b'(t).\label{eq:gtdef}
\end{align}
\label{lem:D1D2ineq}
\end{lemma}

We will also need the following lemma, proved in Appendix~\ref{sec:pfLemmaf}.

\begin{lemma}
	Let $\Phi(t)$ and $D(\rho,\delta)$ be as defined in~\eqref{eq:PhiDef} and~\eqref{eq:DdeltaDef}, respectively. 
	For every $\rho>1$ and $0<\delta<1/2$ for which $D(\rho,\delta)>0$, it holds that
	\begin{align*}
	f(\rho,\delta)\triangleq \frac{1}{\rho}\frac{\Phi(\delta)}{\Phi(D(\rho,\delta))}<1. 
	\end{align*}
	\label{lem:fdec}
\end{lemma}

Finally, we will also need the next proposition, which is a simple variation of the source-channel separation theorem for spherical noise. The proof is given in Appendix~\ref{app:sphericalrdf}
\begin{proposition}
	For any encoder/decoder pair $\m{D}$, $\m{E}$ and any $k\in[-n\delta,n(1-\delta)]$ it holds that
	\begin{align}
	\Psi(k)\geq D\left(\rho,\delta+\frac{k}{n}\right)+\m{O}\left(\frac{\log n}{n}\right),
	\end{align}
	\label{prop:spheseparation} 
	where $\Psi(\cdot)$ was defined in \eqref{eq:PsiDef}.
\end{proposition}

\vspace{2mm}

Using Lemma~\ref{lem:D1D2ineq}, Lemma~\ref{lem:fdec}, and Proposition~\ref{prop:spheseparation}, we can now prove Theorem~\ref{thm:Tradeoff}, which is the main result of this subsection.

\begin{proof}[Proof of Theorem~\ref{thm:Tradeoff}]
First note that for $\epsilon>0$, we can approximate $D(\rho,\delta+\varepsilon)$ as
	\begin{align}
	D(\rho,\delta+\varepsilon)=D(\rho,\delta)+\epsilon D'(\rho,\delta)+\m{O}(\varepsilon^2),\label{eq:DrLimit}
	\end{align}
	where
	\begin{align}
	D'(\rho,\delta) \triangleq \frac{\partial}{\partial \delta}D(\rho,\delta)=\rho\frac{h_b'(\delta)}{h_b'(D(\rho,\delta))}=\rho\frac{\log\left(\frac{1-\delta}{\delta}\right)}{\log\left(\frac{1-D(\rho,\delta)}{D(\rho,\delta)}\right)}.\label{eq:DrDev}
	\end{align}
	Taking $\epsilon=\frac{a}{\sqrt{n}}$, this implies that
\begin{subequations}
\begin{align}
		D(\rho,\delta) - D(\rho,\delta_1) &= \frac{a}{\sqrt{n}} D'(\rho,\delta) + \m{O}\left(\frac{1}{n}\right) \label{eq:Taylor_1} \\
			D(\rho,\delta_1*\delta_2) - D(\rho,\delta) &= \frac{a}{\sqrt{n}} D'(\rho,\delta) + \m{O}\left(\frac{1}{n}\right). \label{eq:Taylor_2}
		\end{align}
		\end{subequations}
		and thus, subtracting~\eqref{eq:Taylor_2} from~\eqref{eq:Taylor_1}, we obtain
		\begin{align} 
			D(\rho,\delta_1*\delta_2)+D(\rho,\delta_1)-2D(\rho,\delta)&=\m{O}\left(\frac{1}{n}\right), \label{eq:DrDiff}
		\end{align}
Now, by Proposition~\ref{prop:spheseparation} we have:
\begin{subequations}
\begin{align}
	 D_1&\geq D(\rho,\delta_1)+\m{O}\left(\frac{\log n}{n}\right) \label{eq:B1}\\
	D_2& \geq D(\rho,\delta_1*\delta_2)+ \m{O}\left(\frac{\log n}{n}\right).  \label{eq:B2}
	\end{align}
	\end{subequations}
Using \eqref{eq:Taylor_2} and \eqref{eq:B2}, 
	we can assert:
	\begin{align} 
	D_2 - D(\rho,\delta) &= (D_2 - D(\rho,\delta_1*\delta_2)) + (D(\rho,\delta_1*\delta_2) - D(\rho,\delta)) \nonumber \\ &\geq \frac{a}{\sqrt{n}}D'(\rho,\delta) + \m{O}\left(\frac{\log n}{n}\right). \label{eq:assert_D2} \end{align} 
Let $\eta=\eta(\rho,\delta)$. We now claim that we can assume without loss of generality:
	\begin{subequations}
	\begin{align}
	 D_1-D(\rho,\delta_1)&<\frac{a\eta}{\sqrt{n}}\label{eq:A1}\\
	D_2-D(\rho,\delta_1*\delta_2)&<\frac{a\eta}{\sqrt{n}}, \label{eq:A2}
	\end{align} \label{eq:A}
	\end{subequations}
To see why this is true, assume to the contrary that one of them, say the first, does not hold. Then, by \eqref{eq:B2} and \eqref{eq:DrDiff},
\begin{align*} D_1+D_2 & \geq D(\rho,\delta_1) + D_2 + \frac{a\eta}{\sqrt{n}} \\
&\geq D(\rho,\delta_1) + D(\rho,\delta_1*\delta_2) + \frac{a\eta}{\sqrt{n}} + \m{O}\left(\frac{\log n}{n}\right) \\
& = 2D(\rho,\delta) + \frac{a\eta}{\sqrt{n}} + \m{O}\left(\frac{\log n}{n}\right) , \end{align*}
which is stronger than the desired bound. 
	
	We now proceed to bound the difference $D_2-D_1$ invoking Lemma~\ref{lem:D1D2ineq} and using \eqref{eq:A}. By the concavity of $t\mapsto h_b(t)$ and \eqref{eq:A2} we have that
	\begin{align}
	h_b(D_2)&=h_b\left(D(\rho,\delta_1*\delta_2)+(D_2-D(\rho,\delta_1*\delta_2))\right)\nonumber\\
	&\leq h_b(D(\rho,\delta_1*\delta_2))+\frac{a\eta}{\sqrt{n}} h_b'(D(\rho,\delta_1*\delta_2)),
	\end{align}
	which implies that
	\begin{align}
	\rho C(\delta_1*\delta_2)-R(D_2)&=\rho C(\delta_1*\delta_2)-\log 2+h_b\left(D_2)\right)\nonumber\\
	&\leq \frac{a\eta}{\sqrt{n}} h_b'(D(\rho,\delta_1*\delta_2))\nonumber\\
	&=\frac{a\eta}{\sqrt{n}}\log\left(\frac{1-D(\rho,\delta_1*\delta_2)}{D(\rho,\delta_1*\delta_2)} \right).
	\end{align}
	Substituting in Lemma~\ref{lem:D1D2ineq}, we have for any $\tau>0$:
	\begin{align}
	D_2-D_1&\leq \frac{1+2D_2\tau}{2D_2\tau} \frac{\frac{a\eta}{\sqrt{n}}\log\left(\frac{1-D(\rho,\delta_1*\delta_2)}{D(\rho,\delta_1*\delta_2)} \right)+\rho\Gamma(n,\delta_2)}{\log\left(\frac{1-D_2}{D_2}\right)}+\frac{2a}{\sqrt{n}}\frac{\log\left(\frac{1-\delta_1}{\delta_1}\right)}{\log\left(\frac{1-D_2}{D_2}\right)}\frac{\Phi(\delta_1)}{\Phi(D_2)}(1+\tau)\frac{g(D_1)}{g(D_2)}\label{eq:D1D2exact}.
	\end{align}
	The functions $t\mapsto \log\left(\frac{1-t}{t}\right)$, $t\mapsto\Phi(t)$, and $t\mapsto g(t)$ are continuous at $0<t< 1/2$. Thus, recalling that $a<\log^2 {(n)}$, by the assumption that $0<D(\rho,\delta)<1/2$, \eqref{eq:A1} and \eqref{eq:A2}, we have that
	\begin{align}
	\frac{\log\left(\frac{1-D(\rho,\delta_1*\delta_2)}{D(\rho,\delta_1*\delta_2)} \right)}{\log\left(\frac{1-D_2}{D_2}\right)}&=1+\m{O}\left(\frac{\log^2{(n)}}{\sqrt{n}}\right)\nonumber\\
	\frac{\log\left(\frac{1-\delta_1}{\delta_1}\right)} {\log\left(\frac{1-D_2}{D_2}\right)}&=\frac{\log\left(\frac{1-\delta}{\delta}\right)}{\log\left(\frac{1-D(\rho,\delta)}{D(\rho,\delta)}\right)}\left(1+\m{O}\left(\frac{\log^2{(n)}}{\sqrt{n}}\right)\right)\nonumber\\
	&= \frac{D'(\rho,\delta)}{\rho} \left(1+\m{O}\left(\frac{\log^2{(n)}}{\sqrt{n}}\right)\right) \nonumber \\
	\frac{\Phi(\delta_1)}{\Phi(D_2)}&=\frac{\Phi(\delta)}{\Phi(D(\rho,\delta))}\left(1+\m{O}\left(\frac{\log^2{(n)}}{\sqrt{n}}\right)\right)\nonumber \\
	&= \rho f(\rho,\delta)  \left(1+\m{O}\left(\frac{\log^2{(n)}}{\sqrt{n}}\right)\right)\nonumber,
	\end{align}
	where $f(\cdot,\cdot)$ is as defined in \eqref{eq:fdef}. In addition, under our assumptions on $\delta_1$ and $\delta_2$, we have that $\Gamma(n,\delta_2)=\m{O}(n^{-3/4}\log{n})$. Thus,~\eqref{eq:D1D2exact} amounts to the following upper bound on the difference $D_2-D_1$:
	\begin{align}
	D_2-D_1\leq &\frac{1}{\sqrt{n}}\left(a\eta\frac{1+2D(\rho,\delta)\tau}{2D(\rho,\delta)\tau}+2a D'(\rho,\delta) f(\rho,\delta)(1+\tau)\right)+\m{O}(n^{-3/4}\log{n}).\label{eq:SlackUB}
	\end{align}

	Combining with~\eqref{eq:assert_D2} now yields
	\begin{align}
	D_1+D_2-2D(\rho,\delta)  
	&= 2(D_2-D(\rho,\delta)) - (D_2-D_1) \nonumber \\ 
	&\geq\frac{2a}{\sqrt{n}} D'(\rho,\delta)  -\frac{1}{\sqrt{n}}\left(a\eta\frac{1+2D(\rho,\delta)\tau}{2D(\rho,\delta)\tau}+2a D'(\rho,\delta)f(\rho,\delta)(1+\tau)\right)
	+\m{O}(n^{-3/4}\log{n})\nonumber\\
	&=\frac{1}{\sqrt{n}}\left(2a D'(\rho,\delta) \left(1-f(\rho,\delta)(1+\tau)\right)-a\eta\frac{1+2D(\rho,\delta)\tau}{2D(\rho,\delta)\tau}\right)+\m{O}(n^{-3/4}\log{n}).\label{eq:D1gap}
	\end{align}
	Now, since $0<f(\rho,\delta)<1$ by Lemma~\ref{lem:fdec}, we can take
	\begin{align}
	\tau=\frac{1-f(\rho,\delta)}{2f(\rho,\delta)}>0,\label{eq:rhoconst}
	\end{align}
	for which~\eqref{eq:D1gap} becomes
	\begin{align}
	D_1+D_2-2D(\rho,\delta)\geq \frac{a\eta(\rho,\delta)}{\sqrt{n}}+\m{O}(n^{-3/4}\log{n}),
	\end{align}
	as desired.
\end{proof}

\section{Proof of Theorem~\ref{thm:main}}
\label{sec:proofmain}

We now turn back to the original point-to-point finite-blocklength JSCC problem. The outline of the proof was already given in Section~\ref{subsec:pfoutline}. Here we give the complete proof. First, without loss of generality we may restrict attention to $\delta\in\left\{\frac{0}{n},\frac{1}{n},\ldots,1\right\}$ as 
\begin{align}
D^*(n,\rho,\delta)\geq D^*\left(n,\rho,\frac{\lfloor n\delta\rfloor}{n}\right),
\end{align}
and we may therefore write
\begin{align}
\Delta^*(n,\rho,\delta)&=D^*(n,\rho,\delta)-D(\rho,\delta)\nonumber\\
&\geq D^*\left(n,\rho,\frac{\lfloor n\delta\rfloor}{n}\right)-D\left(\rho,\frac{\lfloor n\delta\rfloor}{n}\right)+\left[D\left(\rho,\frac{\lfloor n\delta\rfloor}{n}\right)-D(\rho,\delta)\right]\nonumber\\
&=\Delta^*\left(n,\rho,\frac{\lfloor n\delta\rfloor}{n}\right)+\m{O}\left(\frac{1}{n}\right),
\end{align}
where the last equality follows since $D\left(\rho,\frac{\lfloor n\delta\rfloor}{n}\right)-D(\rho,\delta)=\m{O}\left(\frac{1}{n}\right)$. 

As in Section~\ref{subsec:pfoutline}, we define the integer-valued random variable $K = w_H(Z^n) - \delta n$, and for a given encoder/decoder pair $(\m{E},\m{D})$ we 
have that
\begin{align}
\mathbb{E}\left[\frac{1}{m}d_H\left(S^m,\m{D}\left(\m{E}(S^m)+Z^n\right)\right)\right] = \mathbb{E}\left[\Psi(K) \ \right],
\label{eq:iterexpt2}
\end{align}
where 
\begin{align}
\Psi(k) =\Psi_{\m{D},\m{E}}(k)\triangleq \mathbb{E}\left[\frac{1}{m}d_H\left(S^m,\m{D}\left(\m{E}(S^m)+U_k^n\right)\right)\right] ,
\end{align}
and $U_k^n$ is uniform over $\m{S}_{\delta n+k,n}$. 
In terms of this function, we have
\[ \Delta^*(n,\rho,\delta) =\mathbb{E}\left[\Psi(K)-D(\rho,\delta)\right]. \]
We proceed by partitioning $K$ into two regimes:
\begin{align*}
\m{K}_1 &= \{K: |K| \leq \sqrt{n}\log^2(n) \} \\
\m{K}_2 &= \{K: \sqrt{n}\log^2(n) < |K| \} ,
\end{align*}
and asserting:
\begin{align}
\Delta^*(n,\rho,\delta) = \sum_{i=1}^2 \Pr\left( K \in \m{K}_i \right) \mathbb{E}\left[\Psi(K)-D(\rho,\delta) | K \in \m{K}_i \right] . \label{eq:ABC} \end{align} 
and lower bound each of the two terms. 

To lower bound the second term it suffices to note that $\Pr\left( K \in \m{K}_2 \right)=\m{O}(\frac{1}{n})$ and that $\Phi(K)-D(\rho,\delta)$ is bounded, such that its the total contribution is at most $\m{O}(\frac{1}{n})$. 
For the first term, define for all natural $k$
\begin{align} \gamma(k) = \frac{\Pr(K=k) - \Pr(K=-k) }{\Pr(K=k) + \Pr(K=-k) }, \end{align} and write:
\begin{align}
\mathbb{E}\left[\Psi(K)-D(\rho,\delta) \  \bigg| \ |K|=k_0\right]&=\Pr(K=k_0 \ | \ |K|=k_0)\left[\Psi(k_0)-D(\rho,\delta)\right]\nonumber\\
&+\Pr(K=-k_0 \ | \ |K|=k_0)\left[\Psi(-k_0)-D(\rho,\delta)\right]\nonumber\\
&=\frac{1+\gamma(k_0)}{2}\left[\Psi(k_0)-D(\rho,\delta)\right]+\frac{1-\gamma(k_0)}{2}\left[\Psi(-k_0)-D(\rho,\delta)\right],\nonumber\\
&=\frac{1+\gamma(k_0)}{2}\left[\left(\Psi(k_0)-D\left(\rho,\delta+\frac{k_0}{n}\right)\right)+\left(D\left(\rho,\delta+\frac{k_0}{n}\right)-D(\rho,\delta)\right)\right]\nonumber\\
&+\frac{1-\gamma(k_0)}{2}\left[\left(\Psi(-k_0)-D\left(\rho,\delta-\frac{k_0}{n}\right)\right)+\left(D\left(\rho,\delta-\frac{k_0}{n}\right)-D(\rho,\delta)\right)\right]\nonumber\\
&=\frac{1+\gamma(k_0)}{2}\left[\left(\Psi(k_0)-D\left(\rho,\delta+\frac{k_0}{n}\right)\right)\right]+\frac{1-\gamma(k_0)}{2}\left[\left(\Psi(-k_0)-D\left(\rho,\delta-\frac{k_0}{n}\right)\right)\right]\nonumber\\
&+\frac{1}{2}\left[\left(D\left(\rho,\delta+\frac{k_0}{n}\right)-D(\rho,\delta)\right)+\left(D\left(\rho,\delta-\frac{k_0}{n}\right)-D(\rho,\delta)\right)\right]\nonumber\\
&+\frac{\gamma(k_0)}{2}\left[\left(D\left(\rho,\delta+\frac{k_0}{n}\right)-D(\rho,\delta)\right)-\left(D\left(\rho,\delta-\frac{k_0}{n}\right)-D(\rho,\delta)\right)\right]. \label{eq:mean_terms}
\end{align}
In order to bound these quantities, we approximate $D(\rho,\delta+t)=D(\rho,\delta)+c_1 t+c_2 t^2+\m{O}(t^3)$, for some $c_1,c_2\in\RR$, such that
\begin{align}
\left(D\left(\rho,\delta+\frac{k_0}{n}\right)-D(\rho,\delta)\right)+\left(D\left(\rho,\delta-\frac{k_0}{n}\right)-D(\rho,\delta)\right) =\m{O}\left(\frac{k_0^2}{n^2}\right),\label{eq:RDFsymdiff}
\end{align}
and that
\begin{align}
\left(D\left(\rho,\delta+\frac{k_0}{n}\right)-D(\rho,\delta)\right)-\left(D\left(\rho,\delta-\frac{k_0}{n}\right)-D(\rho,\delta)\right)=\m{O}\left(\frac{k_0}{n}\right). 
\end{align}
In addition, in Appendix~\ref{app:pfBinDev} we give Lemma~\ref{lem:bindev}, showing that\footnote{A similar result can be shown for any distribution with a finite third moment using bounds on the (unsigned) Gaussian approximation error such as a theorem by Essen which appears in \cite[Theorem 5.22]{Petrov}. However, we prefer to present the explicit calculation for the binomial distribution.}  for $k_0 < n^{2/3}$,
\begin{align} \gamma(k_0) = \m{O} \left( \max \left( \frac{k_0}{n},\frac{k_0^3}{n^2}  \right)\right) , \label{eq:bindev_order} \end{align}
which amounts to 
\begin{align}
\gamma(k_0)=\m{O}\left(\frac{\log^6{(n)}}{\sqrt{n}}\right), \ 
\forall k_0 \leq \sqrt{n}\log^2(n).\label{eq:gammaapprox}
\end{align}
Applying \eqref{eq:RDFsymdiff}-\eqref{eq:gammaapprox} under the condition $k_0 \leq \sqrt{n}\log^2(n)$, and noting that by Proposition~\ref{prop:spheseparation}  \[ \Psi(k_0)\geq D\left(\rho,\delta+\frac{k_0}{n}\right)+\m{O}\left(\frac{\log n}{n}\right), \] we see that \eqref{eq:mean_terms} amounts to: 
\begin{align}
\mathbb{E}\left[\Psi(K)-D(\rho,\delta) \  \bigg| \ |K|=k_0\right]&\geq
\frac{1-|\gamma(k_0)|}{2}\left[\left(\Psi(k_0)-D\left(\rho,\delta+\frac{k_0}{n}\right)\right)+\left(\Psi(-k_0)-D\left(\rho,\delta-\frac{k_0}{n}\right)\right)\right]\nonumber\\
&+\m{O}\left(\frac{\log^8 n}{n}\right).\label{eq:termB}
\end{align} 
Now, applying~\eqref{eq:RDFsymdiff} and~\eqref{eq:gammaapprox} again,  we see that~\eqref{eq:termB} can be further bounded as 
\begin{align}
\mathbb{E}\left[\Psi(K)-D(\rho,\delta) \  \bigg| \ |K|=k_0\right]
&\geq\frac{1}{2}\left[\left(\Psi(k_0)+\Psi(-k_0)-2D\left(\rho,\delta\right)\right)\right]\left( 1+ \m{O}\left( \frac{\log^6{(n)}}{\sqrt{n}} \right)\right)+\m{O}\left(\frac{\log^8 n}{n}\right).
\end{align}
Letting $a=k_0/\sqrt{n}$ we have that by definition of $\underline{D} (\cdot)$ \eqref{eq:DaDef}, $\Psi(k_0)+\Psi(-k_0) \geq \underline{D}(a)$. 
Thus, the contribution of $\m{K}_1$ in \eqref{eq:ABC} is at most:
\begin{align} 
\Pr \left( K \in \m{K}_1  \right)  \mathbb{E} \left[ \underline D\left( \frac{K}{\sqrt{n}} \right) - 2D\left(\rho,\delta\right) \ \bigg|  \ K \in \m{K}_1 \right] \left( 1+ \m{O}\left( \frac{\log^6{(n)}}{\sqrt{n}} \right)\right) + \m{O}\left(\frac{\log^8 n}{n}\right).
\label{eq:DeltaBoundInter}
\end{align}
Now, using Theorem~~\ref{thm:Tradeoff}, we have that for any fixed $0<a<\log^2{(n)}$,
\begin{align}
 \underline{D}(a)\geq 2D(\rho,\delta)+\frac{a}{\sqrt{n}}\eta(\rho,\delta)+\m{O}(n^{-3/4}\log{n}) 
\end{align}
where $\eta(\rho,\delta)$ is as defined in~\eqref{eq:etaDef}. Substituting in~\eqref{eq:DeltaBoundInter}, we are left with:
\begin{align}
&  \frac{\eta(\rho,\delta)}{\sqrt{n}} \Pr \left( K \in \m{K}_1  \right)  \mathbb{E} \left[ \frac{|K|}{\sqrt{n}}  \bigg|  \ K \in \m{K}_1 \right] \left( 1+ \m{O}\left( \frac{\log^6{(n)}}{\sqrt{n}} \right)\right) + \m{O}\left(\frac{\log^8 n}{n}\right)+\m{O}(n^{-3/4}\log{n}) \nonumber \\ & = \frac{\sqrt{\delta(1-\delta)}\eta(\rho,\delta)}{\sqrt{n}} \Pr \left(|W| < \frac{\log^2{(n)}}{\sqrt{\delta(1-\delta)}}\right)  \mathbb{E} \left[ |W| \bigg|  |W| < \frac{\log^2{(n)}}{\sqrt{\delta(1-\delta)}} \right]  \left( 1+ \m{O}\left( \frac{\log^6{(n)}}{\sqrt{n}} \right)\right) + \m{O}(n^{-3/4}\log{n}) , \end{align}
where the random variable $W=\frac{1}{\sqrt{\delta(1-\delta)}}\frac{K}{\sqrt{n}}$ has zero mean and unit variance. Now let $W_G$ be a standard Gaussian. By using non-uniform bounds on the rate of convergence in the central limit theorem such as the the theorem of Bikelis (see \cite[5.10.4]{Petrov}) , we have that
\begin{align}
\Pr \left(|W| < \frac{\log^2{(n)}}{\sqrt{\delta(1-\delta)}}\right)  \mathbb{E} \left[ |W| \bigg|  |W| < \frac{\log^2{(n)}}{\sqrt{\delta(1-\delta)}} \right] &= \Pr \left(|W_G| < \frac{\log^2{(n)}}{\sqrt{\delta(1-\delta)}}\right)  \mathbb{E} \left[ |W_G| \bigg|  |W_G| < \frac{\log^2{(n)}}{\sqrt{\delta(1-\delta)}} \right]\nonumber\\
& + \m{O} \left( \frac{1}{\sqrt{n}} \right)\nonumber\\
&=\sqrt{\frac{2}{\pi}}+\m{O} \left( \frac{1}{\sqrt{n}} \right) . 
\end{align}
Thus we obtained
\begin{align}
\Delta^*(n,\rho,\delta)\geq \sqrt{\frac{\delta(1-\delta)}{2\pi n}}\eta(\rho,\delta)+\m{O}(n^{-3/4}\log{n}),
\end{align}
as desired.

%

\section{Auxiliary Results: Degraded Memoryless Channels}
\label{sec:aux}

We now consider the special case where the channel $Q^n$ is degraded and memoryless. i.e., $Q^n_{Y_1^n,Y_2^n|X^n}(y_1^n,y_2^n|x^n)=\prod_{i=1}^{n}Q_{Y_1|X}(y_{1i}|x_i)Q_{Y_2|Y_1}(y_{2i}|y_{1i})$. Although this case is not required for our main result, we bring it as a demonstration of the power of Theorem~\ref{thm:JSCCgen}.

We start with the the following lemma, which can be viewed as a restatement of the degraded broadcast channel converse theorem.
\begin{lemma} \label{lem:tensor}
If $Q^n$ is a degraded memoryless broadcast channel, the function $G_{Q^{\otimes n}}$ tensorizes, i.e.,
	\begin{align*}
	G_{Q^{\otimes n}}(nt)  =nG_Q (t).
	\end{align*}
\end{lemma}

\begin{proof}
	We use induction. For any $(W,X^n)$ satisfying the Markov chain $W-X^n-Y_1^n-Y_2^n$ we have
	\begin{align}
	&I(X^n;Y_1^n|W)
	=I(Y_1^{n-1};X^{n-1}|W)+I(X_n;Y_{1,n}|W,Y_1^{n-1}).\nonumber
	\end{align}
	Consequently,
	\begin{align}
	&G_Q\left(\frac{I(X^n;Y_1^n|W)}{n}\right)\nonumber\\
	&=G_Q\left(\frac{I(Y_1^{n-1};X^{n-1}|W)+I(X_n;Y_{1,n}|W,Y_1^{n-1})}{n}\right)\nonumber\\
	&=G_Q\left(\frac{n-1}{n}\frac{I(Y_1^{n-1};X^{n-1}|W)}{n-1}+\frac{I(X_n;Y_{1,n}|W,Y_1^{n-1})}{n}\right)\nonumber\\
	&\geq \frac{n-1}{n}G_Q\left(\frac{I(Y_1^{n-1};X^{n-1}|W)}{n-1}\right)&\nonumber\\
	&+\frac{1}{n}G_Q\left(I(X_n;Y_{1,n}|W,Y_1^{n-1}) \right)\label{eq:Gconcave}
	\end{align}
	where we have used the concavity of $t\mapsto G_Q(t)$ in the last inequality. Invoking the induction hypothesis, we get
	\begin{align}
	&G_Q\left(\frac{I(X^n;Y_1^n|W)}{n}\right)\nonumber\\
	&=\frac{G_{Q^{\otimes(n-1)}}\left(I(Y_1^{n-1};X^{n-1}|W)\right)+G_Q\left(I(X_n;Y_{1,n}|W,Y_1^{n-1})\right)}{n}\nonumber\\
	&\geq \frac{I(Y_2^{n-1},W)+I(Y_{2,n};W,Y_1^{n-1})}{n},\label{eq:Gmark}
	\end{align}
	where the last inequality follows from the definition of $G_{Q^{\otimes(n-1)}}(t)$ and $G_Q(t)$ and the fact that $W-X^n-Y_1^{n-1}-Y_2^{n-1}$ and $(W,Y_1^{n-1})-X_n-Y_{1,n}-Y_{2,n}$ are indeed Markov chains.
	Note that we have the Markov chain $Y_{2,n}-(W,Y_1^{n-1})-Y_2^{n-1}$, and therefore
	\begin{align}
	I(Y_{2,n};W,Y_1^{n-1})&\geq I(Y_{2,n};W,Y_2^{n-1})\nonumber\\
	&\geq I(Y_{2,n};W|Y_2^{n-1}).\label{eq:IY2ineq}
	\end{align}
	Substituting~\eqref{eq:IY2ineq} into~\eqref{eq:Gmark} gives
	\begin{align}
	nG_Q\left(\frac{I(X^n;Y_1^n|W)}{n}\right)&\geq I(Y_2^{n-1},W)+I(Y_{2,n};W|Y_2^{n-1})\nonumber\\
	&=I(Y_2^n;W),\nonumber
	\end{align}
	such that $G_{Q^{\otimes n}}(nt)\leq n G_{Q}(t)$. On the other hand,
	\begin{align} 
G_{Q^{\otimes n}}(nt)& \triangleq\max_{\substack{ {W,X^n \ : \ W-X^n-(Y^n_1,Y^n_2)} \\ { I(X^n;Y^n_1|W)\geq nt}      } }
I(Y^n_2;W) \nonumber\\
&\geq \max_{\substack{ {W^n,X^n \ : \ W^n-X^n-(Y^n_1,Y^n_2)} \\ { I(X^n;Y^n_1|W^n)\geq nt}      } }
I(Y^n_2;W^n) \label{eq:Gqmem}\\
&=n G_Q(t),
\end{align}
where the maximization in~\eqref{eq:Gqmem} is with respect all i.i.d. $(W^n,X^n)\sim P_{WX}^{\otimes n}$.
	%
\end{proof}

The following corollary is an immediate consequence of Theorem~\ref{thm:JSCCgen} and Lemma~\ref{lem:tensor}.

\begin{corollary} 
	Consider the degraded memoryless JSCC broadcast problem. If $(D_1,D_2)$ is achievable, then for any $P=P_S P_{U|S}$, defined by a
	choice of an auxiliary channel $P_{U|S}$,
	\begin{align}
	\bar R_P(D_2) \leq \rho \cdot G_Q\left(\frac{F_P\left(R(D_1) \right) }{\rho} \right).\label{eq:JSCCboundDegMem}
	\end{align}
	\label{thm:JSCCdegraded}
\end{corollary}

This bound can be obtained as a special case of~\cite[Theorem 5]{kc15} (see also~\cite{kc16}), by observing that the boundary of the degraded memoryless broadcast channel $Q$ (without common message) is given by $(C_1,G_Q(C_1))$~\cite[Theorem 5.2]{ElGamalKim}.\footnote{In fact, the techniques developed in~\cite{kc15} should suffice to establish our Theorem~\ref{thm:JSCCgen}. We nevertheless found it more convenient to prove the theorem using properties of the general functions $F_P(t)$, $\bar{R}_P(D)$, and $G_{Q^n}(t)$, as those functions have a major role in other problems in network information theory, see~\cite{kop18} for more details.}

It is not difficult to see that the separation bounds for the extreme cases where only one distortion is of interest are obtained by setting $U=\emptyset$ or $U=S$ for $D_1$ and $D_2$, respectively. 
When $\rho=1$ and the ``not to code'' conditions \cite{grv03} hold, these choices give the tightest bound possible. Otherwise, other choices can give tighter bounds, as demonstrated in the examples below.

\subsection{Quadratic Gaussian Case}

Let $S\sim\m{N}(0,\sigma^2)$, and $d(S_j,\hat S_j) = (S_j - \hat S_j)^2$. We choose $U$ that is the output of an AWGN channel with input $S$ and noise that is Gaussian $(0,\delta^2)$. Using the EPI, one can verify that the corresponding source functions satisfy
\begin{align*}
F_P(t) &= t-\frac{1}{2}\log\left(\frac{\delta^2+\sigma^2}{\delta^2+\sigma^2 e^{-t}}\right) \\
\bar R_P(D) &= \frac{1}{2}\log\left(\frac{\delta^2+\sigma^2}{\delta^2+D}\right), \end{align*}
where $F_P$ is attained by taking $V$ that is the output of an AWGN with input $S$. Furthermore, let $Q^n$ be the (memoryless degarded) AWGN broadcast channel, $Y_1=X+Z_1$, $Y_2=Y_1+Z_2$, where $Z_1\sim\m{N}(0,N_1)$, $Z_2\sim\m{N}(0,N_2)$, $(X,Z_1,Z_2)$ mutually independnet, where the channel input is subject to a quadratic cost constraint $P$. Using the EPI again, one can verify that
\[ G_Q(t)=\frac{1}{2}\log\left(\frac{P+N_1+N_2}{N_1 e^{2t}+N_2} \right), \]
where the function is attained by $(W,X)$ that are jointly Gaussian. Combining with the source functions above and with the quadratic-Gaussian rate-distortion function, and applying Corollary~\ref{thm:JSCCdegraded}, we recover the Reznic, Feder, Zamir outer bound~\cite[Theorem 1]{rfz06}: For all $\delta$,
\[
\frac{\delta^2+\sigma^2}{\delta^2+D_2}\leq \left(1+\frac{P}{N_1+N_2}\right)^{\rho}\left[\frac{N_1+N_2}{N_1\left(\frac{\sigma^2}{D1}\cdot\frac{\delta^2+D_1}{\delta^2+\sigma^2}\right)^{\frac{1}{\rho}}+N_2} \right]^\rho. \]

\subsection{Binary-Hamming Case}
\label{subsec:binary}

We now address the case where $S$ is a Bernoulli($p$) source, and $d(S_j, \hat S_j)$ is the Hamming distortion measure.

ֿ\subsubsection{The Source Functions}

We define $P_{U|S}$ by taking $U=S+ N$, where $N\sim\Ber(q)$ is 
independent of $S$.

\begin{proposition}
	For $0\leq t\leq h_b(p)$
	\begin{align}
	F_P(t) & \geq t-h_b(q*p)+h_b\left(q*h^{-1}\left(h_b(p)-t\right) \right),
	\end{align}
	with equality for $p=1/2$.
	\label{prop:Fbinary}
\end{proposition}

\begin{proof}
	By the Markov structure, we have that $I(S;V)=I(U;V)+I(S;V|U)$. Thus,
	\begin{align*}
	I(S&;V|U) = I(S;V) - H(U) + H(U|V) \\
	&\geq I(S;V) - H(U) + h_b(q*h^{-1}(H(S|V))) \\
	&= I(S;V) - H(U) + h_b(q*h^{-1}(H(S)-I(S;V))) \\
	&= I(S;V) - h_b(q*p) + h_b(q*h^{-1}(h_b(p)-I(S;V))),
	\end{align*}
	where the inequality follows from Mrs. Gerber's Lemma~\cite
	{wz73}. Note that equality holds iff $H(S|V=v)=H(S|V)$ for all $v\in\m{V}$, which is the case for $p=1/2$ and $V=S+ A$, where $A\sim\Ber(h_b^{-1}(1-I(S;V)))$.  \end{proof}


\begin{proposition}
	For $0\leq D\leq p$
	\begin{align}
	\bar{R}_P(D)=h_b(q*p)-h_b(q*D).
	\end{align}
	\label{prop:RbarBin}
\end{proposition}

\begin{proof}
	For every $P_{\hat{S}|S}$ satisfying the constraint $\mathbb{E}(S+\hat{S})\leq D$, we must have that
	\begin{align}
	&I(U;\hat{S})=H(U)-H(U|\hat{S})\nonumber\\
	&=H(U)-H(U+\hat{S}|\hat{S})\nonumber\\
	&\geq H(U)-H(U+\hat{S})\nonumber\\
	&= H(U)-H(N+ S+ \hat{S})\nonumber\\
	&= h_b(q*p)-h_b(q*\mathbb{E}(S+\hat{S}))\nonumber\\
	&\geq h_b(q*p)-h_b(q*D).\nonumber
	\end{align}
	To see that this lower bound is tight, take the reverse test channel $S=\hat{S}+ V$ where $V\sim\Ber(D)$ is statistically independent of $(\hat{S},N)$.
\end{proof}

\subsubsection{Erasure Channel}

Consider first the case where $Q^n$ is a (memoryless degraded) erasure broadcast channel, i.e., $Y_i$ is $X$ w.p. $1-\epsilon_i$ and erased otherwise, for $i=1,2$, where $\epsilon_2\geq\epsilon_1$, and the source is i.i.d. Bernoulli ($p$), and the Hamming distortion measure is used. One can verify that:
\[ G_Q(t) = \frac{1-\epsilon_2}{1-\epsilon_1}(\log 2-\epsilon_1-t). \]
Combining with Propositions \ref{prop:Fbinary} and \ref{prop:RbarBin} and substituting in Corollary~\ref{thm:JSCCdegraded}, one obtains the bound (for $p=1/2$):
\[ \frac{\log2 -h_b(D_2*q)}{(1-\epsilon_2)\log2} + \frac{h_b(D_1*q)-h_b(D_1)}{(1-\epsilon_1)\log2} \leq \rho, \]
which recovers the bound of \cite{tks13} (which was also recovered in~\cite{kc16}).

\subsubsection{Binary Symmetric Channel}

Next, we consider the (memoryless degraded) binary symmetric channel, $Y_1=X+ Z_1$, and $Y_2=Y_1+Z_2$, where $Z_1\sim\Ber(\delta_1)$, $Z_2\sim\Ber(\delta_2)$, and $(X,Z_1,Z_2)$ are mutually independent.

\begin{proposition}
	For the binary symmetric degraded channel
	\begin{align}
	G_Q(t)=\log 2-h_b\left(\delta_2*h_b^{-1}\left(h_b(\delta_1)+t\right)\right),
	\end{align}
	for $0\leq t\leq \log2-h_b(\delta_1)$.
	\label{prop:Gbin}
\end{proposition}

\begin{proof}
	For any $(W,X)$ satisfying the Markov chain $W-X-(Y_1=X+Z_1)-(Y_2=Y_1+Z_2)$, we have
	\begin{align*}
	H(&Y_2|W) \geq h_b\left(\delta_2*h_b^{-1}\left(H(Y_1|W)\right)\right)\\
	&=h_b\left(\delta_2*h_b^{-1}\left(H(Y_1|X)+H(Y_1|W)-H(Y_1|X,W)\right)\right)\\
	&=h_b\left(\delta_2*h_b^{-1}\left(H(Y_1|X)+I(X;Y_1|W)\right)\right)\\
	&=h_b\left(\delta_2*h_b^{-1}\left(h_b(\delta_1)+I(X;Y_1|W)\right)\right),
	\end{align*}
	where the inequality stems from Mrs. Gerber's Lemma and the fact that $Y_2=Y_1+ Z_2$, with equality if $X\sim\Ber(1/2)$ and $W=X+ A$ for $A\sim\Ber(\eta)$, where $I(Y_2;W)=\log 2-h_b(\eta*\delta_1*\delta_2)$. Noticing that $I(Y_2;W)=H(Y_2)-H(Y_2|W) \leq \log 2 - H(Y_2|W)$, with equality for $X\sim\Ber(1/2)$, the proof is completed.\end{proof}

We can now combine this result with Propositions \ref{prop:Fbinary} and \ref{prop:RbarBin} and substitute in Corollary~\ref{thm:JSCCdegraded}, to obtain the following theorem.

\begin{theorem}
	For the JSCC broadcast problem with a $\Ber(p)$ source, Hamming distortion and a binary symmetric channel, suppose that the pair $(D_1,D_2)$ is achievable. Then, for any $0\leq q\leq 1/2$, it holds that
	\[
	h_b(q*p)-h_b(q*D_2) \leq \rho\left[\log 2-h_b\left(\delta_2*h_b^{-1}(A_1)\right)\right], \]
	where
	\begin{align}
	A_1 = h_b(\delta_1)+\frac{1}{\rho}\left[h(q*D_1)-h(D_1)-h(q*p)+h(p)\right].
	\nonumber
	\end{align}
	\label{thm:JSCCbin}
\end{theorem}

For $p=1/2$, the bound significantly simplifies as on the left hand side $h_b(q*p) = \log 2$, while on the right hand side
\begin{align}
A_1 = h_b(\delta_1)+\frac{h_b(q*D_1)-h_b(D_1)}{\rho}.
\end{align}

Following the treatment of the Gaussian-quadratic case in~\cite{rfz06}, we consider the case where the distortion of the ``weak'' user is optimal. That is, let $D^*_2=D(\rho,\delta_1*\delta_2)$, where the function $D(\rho,\delta)$ is as defined in~\eqref{eq:DdeltaDef}.
For the special case of $D_2=D^*_2$. We can take $q\to 0$ in Theorem~\ref{thm:JSCCbin}, and applying some straightforward algebra, we obtain the following.
\begin{corollary}
	For the JSCC broadcast problem with a binary source and a binary symmetric channel, suppose that the pair $(D_1,D^*_2)$ is achievable, where $D^*_2=D(\rho,\delta_1*\delta_2)$. Then,
	\begin{align}
	g(D_1)\geq g(p)+\frac{g(\delta_1)}{g(\delta_1*\delta_2)}\left[g(D^*_2)-g(p)\right] ,
	\label{eq:JSCCbinPD2}
	\end{align}
	where $g(t)\triangleq(1-2t)\log\left(\frac{1-t}{t}\right)$.
	\label{cor:JSCCbinD2opt}
\end{corollary}

Similarly, for the special case of $D_1=D_1^*=D(\rho,\delta_1)$, we can take $q\to1/2$ in Theorem~\ref{thm:JSCCbin}, and after applying some straightforward algebra, obtain the following.
\begin{corollary}
	For the JSCC broadcast problem with a binary source and a binary symmetric channel, suppose that the pair $(D_1^*,D_2)$ is achievable, where $D^*_1=D(\rho,\delta_1)$. Then,
	\begin{align}
	(1-2D_2)^2&\leq (1-2\cdot \delta_2*D_1^*)^2\nonumber\\
	&+(1-2p)^2\left(1-(1-2\cdot\delta_2*D_1^*)^2\right).
	\label{eq:JSCCbinPD1}
	\end{align}
	In particular, for $p=1/2$,
	\begin{align}
	D_2\geq \delta_2*D_1^*.
	\end{align}
	\label{cor:JSCCbinD1opt}
\end{corollary}

\section{Discussion: The Remaining Gap to Achievable Performance}
\label{sec:discuss}

In this work we have shown an example, where $\Delta^*_n = \Omega \left( n^{-1/2} \right)$. It is natural, of course, to ask whether such performance is also achievable. 

Consider a separation-based scheme: the source is quantized to a rate $R_n$ with expected distortion $D_{0,n}$. This code is matched to a channel code with the same rate. Upon correct channel decoding we have distortion $D_{0,n}$, while incorrect decoding gives disortion that is trivially upper bounded by $1$. If the channel error probability is $p_n$, this scheme yields
\begin{align} D_n \leq ( 1 - p_n) D_{0,n} + p_n. \end{align}  
Now, we know that for the lossy source problem it is possible to achieve
\[ D_{0,n} \leq D(\rho R_n) + \m{O} \left( \frac{\log n}{n} \right), \] 
 Thus, 
 \begin{align} 
 \Delta_n &\leq D_{0,n} - D^*_\infty + p_n \left[1 -D_{0,n}\right] \nonumber \\
 &=  D(\rho R_n) - D(\rho C) + \m{O} (p_n) +  \m{O} \left( \frac{\log n}{n} \right) \nonumber \\
 &=  \m{O}( C - R_n ) + \m{O} (p_n) +  \m{O} \left( \frac{\log n}{n} \right) .
 \end{align} 
In order for both the first term and the second term to decrease, we must choose $R_n$ in the moderate-deviations regime. Using \cite{AW_moderate,PV_moderate} we have that:
 \[ \log p_n = \m{O}( n (C - R_n )^2). \] 
Substituting, we find that a separation-based scheme achieves
\begin{align} \Delta_n = \m{O} \left( \sqrt\frac{\log n}{n} \right) . \end{align}

Next, one can consider the combination of successive-refinement (SR) source coding with a digital channel broadcast code, possibly with many layers to track well the channel quality, as done in the context of Gaussian channels in different formulations regarding the high signal-to-noise ratio regime \cite{BroadcastApproach,TaherzadehKhandani,cn07}. However, one may verify that these techniques will not improve upon the order of convergence of a separation-based scheme due to the following consideration. The first layer of the SR code will have to be allocated a rate that is the same order below capacity as in the separation scheme. In order to reduce distortion, we will need a layer that will be correctly decoded when the empirical channel is above capacity in the same $\sqrt{\log n / n}$ order. But substituting in the broascast channel converse, it turns out that this refinement layer will be able to carry a very low rate, failing to reduce the distortion by the required amount.

We see that the limitation of digital schemes stems from two effects.
\begin{enumerate}
\item A broadcast code for two empirical channels that are symmetric around capacity, carries a sum-rate that is much lower than capacity.
\item The threshold effect: a digital JSCC scheme that performs well around capacity, cannot accommodate for bad channel conditions. 
\end{enumerate}
Indeed, recalling that our converse bound is based upon bounding the channel function $G_Q$ that is intimately related to the converse for the digital BC problem, it can be seen as reflecting the first effect. However, it does not reflect the second, hence the remaining gap between the achievable and converse bounds,  It remains to be seen, whether this threshold effect indeed applies to all relevant schemes.

One
indication that the threshold effect might be unavoidable are the
results of~\cite{rp18b}, which show that linear codes with any non-zero minimal
distance necessarily admit medium-sized error vectors that result in
maximal Hamming distortion for the input bits. One approach for finding codes with
graceful-degradation (or adaptation to channel conditions) was
suggested in~\cite{kmp12}, see~\cite{p17,rp18a} for more results.

As a possible coding scheme to improve upon separation, we can consider coding for the simplest case - a binary erasure channel (BEC). In this case, we only
need to produce a good lossy source code which produces from $m$ (source) bits $n$ coded bits with the following
``fountain-code-like'' property: if any subset of $mR(D)$ of
these coded bits is available (i.e. these positions are not erased by the BEC) then the source can be reconstructed
with distortion $mD$, and this property should hold for a small
range of $D = D_{\infty}^* \pm {c/\sqrt{n}}$. Since the channel returns $nC +
\sqrt{nV} Z$ unerased bits ($Z\sim \mathcal{N}(0,1)$), by averaging we would get distortion
$D_{\infty}^*(1+o(1/\sqrt{n}))$. So the only remaining task is to construct
this ``fountain-like'' rate-distortion code. Whether such a code exists is an open problem, although known results about multiple-description
problem (see, e.g., \cite{Chen09} in a Gaussian setting) suggest that the required property is not possible for all subsets of coded bits (without restriction on their size).

\begin{appendices}
	
\section{Proof of Proposition~\ref{prop:spheseparation}}
\label{app:sphericalrdf}

	By the definition of $\Psi(k)$ and the binary symmetric rate-distortion function, we have that for any encoder and decoder pair $(\m{E},\m{D})$ it holds that
	\begin{align}
	I\left(S^m;\m{D}\left(\m{E}(S^m)+U_k^n\right)\right)\geq m\left(\log 2-h_b(\Psi(k))\right).\label{eq:PsiRD}
	\end{align}
	Using the data processing inequality we obtain
	\begin{align}
	I\left(S^m;\m{D}\left(\m{E}(S^m)+U_k^n\right)\right)\leq I\left(\m{E}(S^m);\m{E}(S^m)+U_k^n\right).\label{eq:PsiDPI}
	\end{align}
	Recalling that~\cite[Chapter 10, Lemma 7]{ms77} 
	\begin{align}
	H(U_k^n)&\geq n h_b\left(\delta+\frac{k}{n}\right)-\frac{\log n}{2}-\frac{1}{2}\log\left(2\pi\left(\delta+\frac{k}{n}\right)\left(1-\delta-\frac{k}{n}\right)\right)\nonumber\\
	&\geq n h_b\left(\delta+\frac{k}{n}\right)-\frac{\log n}{2}-\frac{1}{2},
	\end{align}
	we have that for any random vector $X^n$ in $\{0,1\}^n$ 
	\begin{align}
	I(X^n;X^n+U_k^n)&\leq H(X^n+U_k^n)-H(U_k^n)\nonumber\\
	&\leq n\log 2-n h_b\left(\delta+\frac{k}{n}\right)+\frac{\log{n}}{2}+\frac{1}{2}.\label{eq:MIsphere}
	\end{align}
	Thus, combining~\eqref{eq:PsiRD},~\eqref{eq:PsiDPI}, and~\eqref{eq:MIsphere}, and recalling that $\rho=n/m$, we have
	\begin{align}
	\log 2-h_b\left(\Psi(k)\right)\leq \rho\left(\log 2-h_b\left(\delta+\frac{k}{n}\right)+\frac{\log{n}+1}{2n}\right),
	\end{align}
	which yields the desired result.

\section{The Asymmetry of the Binomial Distribution}
\label{app:pfBinDev}

\begin{lemma}
	For $n\delta\in\ZZ$ define the integer random variable $K=\Bin(n,\delta)-n\delta$. For any integer $0\leq k<n^{2/3}$ it holds that
	\begin{align}
	\frac{\Pr\left(K=k\right)}{\Pr\left(K=k\right)+\Pr\left(K=-k\right)}=\frac{1}{2}+\frac{1}{4}\frac{(1-2\delta)}{\delta(1-\delta)}\left[\frac{k^2}{n}\frac{1}{3\delta(1-\delta))}-1\right]\frac{k}{n}+o\left(\frac{k^3}{n^2} \right).\nonumber
	\end{align}
	\label{lem:bindev}
\end{lemma}

\begin{proof}
	Using the binomial distribution, we have
	\begin{align}
	\Gamma&\triangleq\frac{\Pr\left(K=k\right)}{\Pr\left(K=-k\right)}\nonumber\\
	&=\frac{{n\choose{n\delta+k}}\left(\frac{\delta}{1-\delta}\right)^{k}}{{n\choose{n\delta-k}}\left(\frac{\delta}{1-\delta}\right)^{-k}}\nonumber\\
	&=\frac{(n\delta-k)!}{(n\delta+k)!}\frac{(n(1-\delta)+k)!}{(n(1-\delta)-k)!}\left(\frac{\delta}{1-\delta}\right)^{2k}\nonumber\\
	&=\frac{(n(1-\delta))^{2k}\prod_{i=1}^{2k}\left(1-\frac{k}{(1-\delta)n}+\frac{i}{n(1-\delta)}\right)}{(n\delta)^{2k}\prod_{i=1}^{2k}\left(1-\frac{k}{\delta n}+\frac{i}{n\delta}\right)}   \left(\frac{\delta}{1-\delta}\right)^{2k}\nonumber\\
	&=\prod_{i=1}^{2k}\frac{1-\frac{k}{(1-\delta)n}+\frac{i}{n(1-\delta)}}{1-\frac{k}{\delta n}+\frac{i}{n\delta}}
	\end{align}
	Letting $b_i=\frac{i}{n}-\frac{k}{n}$, we have
	\begin{align}
	\log\Gamma=\sum_{i=1}^{2k}\log(1+\frac{1}{1-\delta}b_i)-\sum_{i=1}^{2k}\log(1+\frac{1}{\delta}b_i).
	\end{align}
	Recalling that $\log(1+x)=x-\frac{x^2}{2}+o(x^2)$, we see that
	\begin{align}
	\log\Gamma=\left(\frac{1}{1-\delta}-\frac{1}{\delta}\right)\sum_{i=1}^{2k}b_i-\frac{1}{2}\left(\frac{1}{(1-\delta)^2}-\frac{1}{\delta^2}\right)\sum_{i=1}^{2k}b_i^2+o\left(\frac{k^3}{n^2}\right).
	\end{align}
	Furthermore, we have
	\begin{subequations}
	\begin{align}
	\sum_{i=1}^{2k}b_i &=\frac{k}{n}+\sum_{j=-(k-1)}^{k-1}\frac{j}{n}=\frac{k}{n} \\
	\sum_{i=1}^{2k}b^2_i &=\frac{k^2}{n^2}+2\sum_{j=0}^{k-1}\frac{i^2}{n^2}=\frac{k^2}{n^2}+\frac{2}{n^2}\frac{(k-1)\cdot k\cdot (2k-1)}{6}=\frac{2}{3}\frac{k^3}{n^2}+o\left(\frac{k^3}{n^2}\right)
	\end{align}
\end{subequations}
	Thus,
	\begin{align}
	\log\Gamma&=-\frac{1-2\delta}{\delta(1-\delta)}\frac{k}{n}+\frac{1-2\delta}{(\delta(1-\delta))^2}\frac{k^3}{3 n^2}+o\left(\frac{k^3}{n^2} \right)\nonumber\\
	&=\frac{(1-2\delta)}{\delta(1-\delta)}\left[\frac{k^2}{n}\frac{1}{3\delta(1-\delta))}-1\right]\frac{k}{n}+o\left(\frac{k^3}{n^2} \right).
	\end{align}
	Since $e^{x}=1+x+\m{O}(x^2)$, it follows that
	\begin{align}
	\Gamma=1+\frac{(1-2\delta)}{\delta(1-\delta)}\left[\frac{k^2}{n}\frac{1}{3\delta(1-\delta))}-1\right]\frac{k}{n}+o\left(\frac{k^3}{n^2} \right).
	\end{align}
	Now, since the required result is $\Gamma / (1+\Gamma)$, it follows easily. 
\end{proof}
	
\section{Auxiliary Lemmas}	
\label{app:aux}

\begin{lemma}[MGL linearization]
For any $0\leq\delta_1,\delta_2\leq 1/2$ and $-h_b(\delta_1)<x<\log{2}-h_b(\delta_1)$ we have that
\begin{align}
h_b\left(\delta_2*h_b^{-1}\left(h_b(\delta_1)+x\right)\right)\geq h_b(\delta_1*\delta_2)+\frac{g(\delta_1*\delta_2)}{g(\delta_1)}x,
\end{align}
where $g(t)=(1-2t)\log\left(\frac{1-t}{t}\right)$, as defined in~\eqref{eq:gtdef}.
\label{lem:MGLlin}
\end{lemma}

\begin{proof}
Let $\varphi_{\delta_2}(t)=h_b\left(\delta_2*h_b^{-1}\left(t\right)\right)$ and recall that by~\cite{wz73} we have that $t\mapsto\varphi_{\delta_2}(t)$ is convex, and 
\begin{align}
\varphi'_{\delta_2}(t)=\frac{d}{dt}\varphi_{\delta_2}(t)=(1-2\delta_2)\frac{h_b'(\delta_2*h_b^{-1}(t))}{h_b'(h_b^{-1}(t))}.
\end{align}	
Consequently,
\begin{align}
\varphi'_{\delta_2}(h_b(\delta_1))&=(1-2\delta_2)\frac{h_b'(\delta_1*\delta_2)}{h_b'(\delta_1)}\\
&=\frac{(1-2\delta_2)(1-2\delta_1)}{(1-2\delta_1)}\frac{h_b'(\delta_1*\delta_2)}{h_b'(\delta_1)}\\
&=\frac{(1-2(\delta_1*\delta_2))}{(1-2\delta_1)}\frac{h_b'(\delta_1*\delta_2)}{h_b'(\delta_1)}\\
&=\frac{g(\delta_1*\delta_2)}{g(\delta_1)}.\label{eq:MGLder}
\end{align}
Now, by convexity, we have
\begin{align}
h_b\left(\delta_2*h_b^{-1}\left(h_b(\delta_1)+x\right)\right)\geq \varphi_{\delta_2}(h_b(\delta_1))+\varphi'_{\delta_2}(h_b(\delta_1))x,\label{eq:RHSbound}
\end{align}
and the statement follows by substituting~\eqref{eq:MGLder} into~\eqref{eq:RHSbound}.
\end{proof}

\begin{lemma}[Properties of $g(t)$]
The function $t\mapsto g(t)=(1-2t)\log\left(\frac{1-t}{t}\right)$ is convex in $[0,1/2]$ and its derivative is given by
\begin{align}
g'(t)=\frac{d}{dt}g(t)=-\kappa(t),
\end{align}
where
\begin{align}
\kappa(t)&\triangleq2\log\left(\frac{1-t}{t}\right)+\frac{1-2t}{t(1-t)}.\label{eq:kappadef}
\end{align}
\label{lem:gconv}
\end{lemma}

\begin{proof}
Calculating $g'(t)$ is straightforward. Furthermore, all three functions $t\mapsto 2\log\left(\frac{1-t}{t}\right)$, $t\mapsto 1-2t$, and $t\mapsto \frac{1}{t(1-t)}$ are decreasing in $[0,1/2]$, so $\kappa(t)$ is decreasing, $g'(t)$ increasing, and $g''(t)>0$.
\end{proof}

\begin{lemma}[Properties of $\beta_q(t)$]
Let $\beta_q(t)\triangleq h_b(q*t)-h_b(t)$. The function $t\mapsto \beta_q(t)$ is convex, and its derivative satisfies
\begin{align}
\beta'_q(t)=\frac{d}{dt}\beta_q(t)=-\phi(q,t),
\end{align}	
where
\begin{align}
\phi(q,t)\triangleq 2q\log\left(\frac{1-t}{t}\right)+(1-2q)\log\left(\frac{1+q\frac{1-2t}{t}}{1-q\frac{1-2t}{1-t}}\right).\label{eq:phidef}
\end{align}
Furthermore, for $t\in[0,1/2]$
\begin{align}
\phi(q,t)\geq q\cdot \kappa(t)\cdot \nu(q,t),
\end{align}
where $\kappa(t)$ is as in~\eqref{eq:kappadef}, and
\begin{align}
\nu(q,t)&\triangleq\frac{1-2q}{1+q\frac{1-2t}t}.\label{eq:PhiIneq}
\end{align}
Moreover,
\begin{align}
\beta_q(t)\leq q\cdot g(t).\label{eq:fqub}
\end{align}
\label{lem:fconv}
\end{lemma}

\begin{proof}
Convexity follows by noting that $\beta_q(t)=I(X;X+Z)$, where $X\sim\Ber(q)$ and $Z\sim\Ber(t)$. Calculation of $\beta'_q(t)$ is straightforward.
In order to lower bound $\phi(q,t)$, we note that
\begin{align}
\log\left(1+q\frac{1-2t}{t}\right)\geq q\frac{\frac{1-2t}{t}}{1+q\frac{1-2t}{t}}\\
-\log\left(1-q\frac{1-2t}{1-t}\right)\geq q\frac{1-2t}{1-t},
\end{align}
and therefore
\begin{align}
(1-2q)\log\left(\frac{1+q\frac{1-2t}{t}}{1-q\frac{1-2t}{1-t}}\right)&\geq q\frac{1-2q}{1+q\frac{1-2t}{t}}\left(\frac{1-2t}{t}+\frac{1-2t}{1-t} \right)\\
&=q\frac{1-2t}{t(1-t)}\frac{1-2q}{1+q\frac{1-2t}{t}},
\end{align}
which gives
\begin{align}
\phi(q,t)&\geq q\left(2\log\left(\frac{1-t}{t}\right)+\frac{1-2t}{t(1-t)}\frac{1-2q}{1+q\frac{1-2t}{t}}\right)\\
&\geq q\underbrace{\left(2\log\left(\frac{1-t}{t}\right)+\frac{1-2t}{t(1-t)}\right)}_{\kappa(t)}\underbrace{\left(\frac{1-2q}{1+q\frac{1-2t}{t}}\right)}_{\nu(q,t)}.\nonumber
\end{align}
To prove~\eqref{eq:fqub}, we apply the concavity of $t\mapsto h_b(t)$ to obtain
\begin{align}
\beta_q(t)&=h_b(q*t)-h_b(t)\nonumber\\
&=h_b(t+q(1-2t))-h_b(t)\nonumber\\
&\leq q(1-2t)h_b'(t)\nonumber\\
&=q g(t).\nonumber
\end{align}
\end{proof}

\section{Proof of Lemma~\ref{lem:D1D2ineq}}
\label{app:D1D2lemma}

By Theorem~\ref{thm:JSCC_Delta}, we have that for any $0<q<1/2$
\begin{align}
\log 2-h_b(q*D_2) &\leq \rho\left[\log 2-h_b\left(\delta_2*h_b^{-1}\left(h_b(\delta_1)+\frac{h_b(q*D_1)-h_b(D_1)}{\rho}\right)\right)\right]+\rho\Gamma(n,\delta_2).\label{eq:basicbound}
\end{align}
Now, applying Lemma~\ref{lem:MGLlin}, we can write
\begin{align}
h_b\left(\delta_2*h_b^{-1}\left(h_b(\delta_1)+\frac{h_b(q*D_1)-h_b(D_1)}{\rho}\right)\right)\geq h_b(\delta_1*\delta_2)+\frac{g(\delta_1*\delta_2)}{g(\delta_1)}\frac{\beta_q(D_1)}{\rho},
\label{eq:MGLbound}
\end{align}
where $\beta_q(\cdot)$ is as defined in Lemma~\ref{lem:fconv}. By combining~\eqref{eq:basicbound} and~\eqref{eq:MGLbound}, we obtain
\begin{align}
\log{2}-h_b(D_2)-[h_b(q*D_2)-h_b(D_2)]\leq \rho(\log{2}-h(\delta_1*\delta_2))-\frac{g(\delta_1*\delta_2)}{g(\delta_1)}\beta_q(D_1)+\rho\Gamma(n,\delta_2),
\end{align}
which, recalling the definition of $C(t)$ and $R(t)$, reduces to
\begin{align}
\rho C(\delta_1*\delta_2)-R(D_2)\geq\frac{g(\delta_1*\delta_2)}{g(\delta_1)}\beta_q(D_1)-\beta_q(D_2)-\rho\Gamma(n,\delta_2).\label{eq:qbound}
\end{align}
By Lemma~\ref{lem:gconv}, the function $t\mapsto g(t)$ is convex, and consequently
\begin{align}
g(\delta_1*\delta_2)=g(\delta_1+(1-2\delta_1)\delta_2)\geq g(\delta_1)+\delta_2\cdot(1-2\delta_1)g'(\delta_1).
\end{align}
Thus, we obtain
\begin{align}
\rho C(\delta_1*\delta_2)-R(D_2)&\geq \beta_q(D_1)-\beta_q(D_2)+\delta_2 \frac{(1-2\delta_1)g'(\delta_1)}{g(\delta_1)} \beta_q(D_1)-\rho\Gamma(n,\delta_2)\\
&=\beta_q(D_1)-\beta_q(D_2)-\psi(\delta_1)\delta_2 \beta_q(D_1)-\rho\Gamma(n,\delta_2),
\end{align}
where
\begin{align}
\psi(t)&\triangleq-\frac{(1-2t)g'(t)}{g(t)}\nonumber\\
&=(1-2t)\frac{\kappa(t)}{g(t)},	
\end{align}
where $\kappa(\cdot)$ is as defined in Lemma~\ref{lem:gconv}. By Lemma~\ref{lem:fconv}, the function $t\mapsto \beta_q(t)$ is convex, and consequently	
\begin{align}
\beta_q(D_1)&=\beta_q(D_2+(D_1-D_2))\nonumber\\
&\geq \beta_q(D_2)+\beta'_q(D_2)(D_1-D_2)\nonumber\\
&= \beta_q(D_2)+\phi(q,D_2)(D_2-D_1),
\end{align}
where $\phi(\cdot,\cdot)$ is as defined in Lemma~\ref{lem:fconv}.
We have obtained
\begin{align}
D_2-D_1\leq\frac{\left[\rho C(\delta_1*\delta_2)-R(D_2)\right]+\psi(\delta_1)\delta_2 \beta_q(D_1)+\rho\Gamma(n,\delta_2)}{\phi(q,D_2)}.
\end{align}
By Lemma~\ref{lem:fconv}, we have that $\beta_q(D_1)\leq qg(D_1)$, and consequently
\begin{align}
D_2-D_1&\leq \frac{\rho C(\delta_1*\delta_2)-R(D_2)+\rho\Gamma(n,\delta_2)}{\phi(q,D_2)}+\frac{\psi(\delta_1)\delta_2 g(D_1)}{\phi(q,D_2)/q}\\
&\leq\frac{\rho C(\delta_1*\delta_2)-R(D_2)+\rho \Gamma(n,\delta_2)}{2q\log\left(\frac{1-D_2}{D_2}\right)}+(\delta_1*\delta_2-\delta_1)\frac{\psi(\delta_1)}{1-2\delta_1}\frac{g(D_2)}{\kappa(D_2)}\frac{1}{\nu(q,D_2)}\frac{g(D_1)}{g(D_2)},\label{eq:GapIneq2}
\end{align}
where in the last inequality we have used the identity $\delta_2=\frac{\delta_1*\delta_2-\delta_1}{1-2\delta_1}$, the fact that $\phi(q,D_2)>2q\log\left(\frac{1-D_2}{D_2}\right)$ by~\eqref{eq:phidef} and that $\phi(q,D_2)>q\cdot\kappa(t)\cdot\nu(q,t)$ by~\eqref{eq:PhiIneq}. 
Let
\begin{align}
\tau\triangleq \frac{q}{(1-2q)D_2},
\end{align}
such that $\frac{1}{\nu(q,D_2)}=1+\tau$, and $q=\frac{D_2\tau}{1+2D_2\tau}$. Since~\eqref{eq:GapIneq2} holds for any $0<q<1/2$, we have that for any $\tau>0$
\begin{align}
D_2-D_1&\leq\frac{\rho C(\delta_1*\delta_2)-R(D_2)+\rho\Gamma(n,\delta_2)}{\frac{2D_2\tau}{1+2D_2\tau}\log\left(\frac{1-D_2}{D_2}\right)}+(\delta_1*\delta_2-\delta_1)\frac{\psi(\delta_1)}{1-2\delta_1}\frac{g(D_2)}{\kappa(D_2)}(1+\tau)\frac{g(D_1)}{g(D_2)}.\label{eq:GapIneq3}
\end{align}
Note that
\begin{align}
\frac{\psi(t)}{1-2t}&=\frac{\kappa(t)}{g(t)}\nonumber\\
&=\frac{2}{1-2t}+\frac{1}{t(1-t)\log\left(\frac{1-t}{t}\right)}\nonumber\\
&=\log\left(\frac{1-t}{t}\right)\Phi(t),
\label{eq:Phi1}
\end{align}
where $\Phi(t)$ is defined in~\eqref{eq:PhiDef}.
Thus,~\eqref{eq:GapIneq2} can be written as
\begin{align}
D_2-D_1&\leq\frac{\rho C(\delta_1*\delta_2)-R(D_2)+\rho\Gamma(n,\delta_2)}{\frac{2D_2\tau}{1+2D_2\tau}\log\left(\frac{1-D_2}{D_2}\right)}+(\delta_1*\delta_2-\delta_1)\rho\frac{\log\left(\frac{1-\delta_1}{\delta_1}\right)}{\log\left(\frac{1-D_2}{D_2}\right)}\frac{1}{\rho}\frac{\Phi(\delta_1)}{\Phi(D_2)}(1+\tau)\frac{g(D_1)}{g(D_2)},
\end{align}
which establishes our claim.

\section{Proof of Lemma~\ref{lem:fdec}}
\label{sec:pfLemmaf}

For any $0<\delta<1/2$ we have that $f(1,\delta)=1$. Thus, it suffices to show that $\rho\mapsto f(\rho,\delta)$ is monotone decreasing. Let $\rho=\rho(d)$ be such that $D(\rho,\delta)=d>0$, which is well defined for $d>0$. It is easy to see that
	\begin{align}
	\rho(d)=\frac{\log 2-h_b(d)}{\log 2-h_b(\delta)}=\frac{R(d)}{R(\delta)}.
	\end{align}
	Recalling that $d<\delta$ for $\rho>1$, our claim is equivalent to the claim that for any $0<d<\delta$ it holds that
	\begin{align}
	\frac{R(d)}{R(\delta)}\Phi(d)> \Phi(\delta),
	\end{align}
	which is equivalent to the claim that
	\begin{align}
	\vartheta(t)\triangleq \Phi(t)\cdot R(t),
	\end{align}
	is monotone decreasing in $t$, which we now establish.
	
	To this end, note that we can write
	\begin{align}
	\vartheta(t)=\frac{\gamma(t)}{g^2(t)},
	\end{align}
	where
	\begin{align}
	\gamma(t)\triangleq (1-2t)R(t)\kappa(t),
	\end{align}
	and $\kappa(t)=-g'(t)$ is as defined in~\eqref{eq:kappadef}. We therefore have that
	\begin{align}
	\vartheta'(t)&=\frac{1}{g^4(t)}\left[g^2(t)\gamma'(t)-2g(t)g'(t)\gamma(t)\right]\nonumber\\
	&=\frac{1}{g^3(t)}\left[\underbrace{g(t)\gamma'(t)+2\kappa(t)\gamma(t)}_{\zeta(t)}\right].
	\end{align}
	Since $g(t)>0$ for all $0<t<1/2$, we have to show that $\zeta(t)\leq 0$. We write
	\begin{align}
	\gamma'(t)&=\kappa(t)(1-2t)R'(t)+\kappa(t)R(t)[(1-2t)']+(1-2t)R(t)\kappa'(t)\nonumber\\
	&=-\kappa(t)g(t)-2\kappa(t)R(t)+(1-2t)R(t)\kappa'(t),
	\end{align}
	where the last equality follows since $R'(t)=-\log\left(\frac{1-t}{t}\right)$, and therefore $(1-2t)R'(t)=-g(t)$. Furthermore, as $(1-2t)\kappa(t)=2g(t)+\frac{(1-2t)^2}{t(1-t)}$, we also that 
	\begin{align}
	\kappa(t)\gamma(t)=2g(t)\kappa(t)R(t)+\frac{(1-2t)^2}{t(1-t)}\kappa(t)R(t).
	\end{align}
	Thus,
	\begin{align}
	\zeta(t)&=-\kappa(t)g^2(t)-2g(t)\kappa(t)R(t)+(1-2t)g(t)R(t)\kappa'(t)+4g(t)\kappa(t)R(t)+2\frac{(1-2t)^2}{t(1-t)}\kappa(t)R(t)\nonumber\\
	&=-\kappa(t)g^2(t)+2g(t)\kappa(t)R(t)+(1-2t)g(t)R(t)\kappa'(t)+2\frac{(1-2t)^2}{t(1-t)}\kappa(t)R(t).\label{eq:zetatmp}
	\end{align}
	Evaluating $\kappa'(t)$ gives
	\begin{align}
	\kappa'(t)&=-\frac{2}{t(1-t)}-\frac{1-2t+2t^2}{(t(1-t))^2}\nonumber\\
	&=-\frac{1}{(t(1-t))^2}.\label{eq:kappader}
	\end{align}
	Substituting~\eqref{eq:kappader} into~\eqref{eq:zetatmp}, gives
	\begin{align}
	\zeta(t)&=-\kappa(t)g^2(t)+2g(t)\kappa(t)R(t)-\frac{R(t)g(t)(1-2t)}{(t(1-t))^2}+2\frac{(1-2t)^2}{t(1-t)}\kappa(t)R(t)\nonumber\\
	&=-g(t)\left[\kappa(t)g(t)-2\kappa(t)R(t)+\frac{R(t)(1-2t)}{(t(1-t))^2}-2\frac{(1-2t)}{t(1-t)\log\left(\frac{1-t}{t}\right)}\kappa(t)R(t)\right]\nonumber\\
	&=-g(t)\left[\kappa(t)g(t)-2\kappa(t)R(t)+\frac{R(t)(1-2t)}{(t(1-t))^2}-\frac{4R(t)(1-2t)}{t(1-t)}-\frac{2R(t)(1-2t)^2}{t^2(1-t)^2\log\left(\frac{1-t}{t}\right)}\right]\nonumber\\
	&=-g(t)\left[\kappa(t)g(t)-2\kappa(t)R(t)+\frac{R(t)(1-2t)^3}{(t(1-t))^2}-\frac{2R(t)(1-2t)^2}{t^2(1-t)^2\log\left(\frac{1-t}{t}\right)}\right]\nonumber\\
	&=-g(t)\left[\underbrace{\kappa(t)g(t)-2\kappa(t)R(t)}_{A_1(t)}+\underbrace{\frac{R(t)(1-2t)^2}{(t(1-t))^2}}_{A_2(t)}
	\underbrace{\left[1-2\left(t+\frac{1}{\log\left(\frac{1-t}{t}\right)}\right)\right]}_{A_3(t)}
	\right].\nonumber
	\end{align}
	It remains to show that the function $A(t)=A_1(t)+A_2(t)A_3(t)$ is positive in $0<t<1/2$. To this end, one can easily verify that $A(0)=\infty$ and $A(1/2)=0$, and therefore, it suffices to verify that $t\mapsto A(t)$ is decreasing in $0<t<1/2$. Indeed, it is straightforward to verify that $t\mapsto A_1(t)$, $t\mapsto A_2(t)$ and $t\mapsto A_3(t)$ are decreasing, and the details are omitted.

\end{appendices}



\bibliographystyle{IEEEtran}
\bibliography{JSCCbib}

\end{document}